\documentclass{article}

\usepackage[preprint]{neurips_2025}
\usepackage[ruled]{algorithm2e} 

\SetAlFnt{\small}
\usepackage{amsthm}
\usepackage{amsmath}
\usepackage{stmaryrd}
\usepackage{dsfont}
\usepackage{aliases}
\usepackage{tikz}
\usepackage{subcaption}
\usepackage{thm-restate}
\newtheorem{remark}{Remark}
\newtheorem{theorem}{Theorem}

\newtheorem*{proposition*}{Proposition}
\newtheorem{lemma}{Lemma}
\newtheorem{corollary}{Corollary}

\newenvironment{sketch}{{\bf \emph{Sketch of proof.} }}{\hfill$\square$}
\usepackage{comment}

\newcommand{\triangleq}{\stackrel{\triangle}{=}}

\usepackage[utf8]{inputenc} 
\usepackage[T1]{fontenc}    
\usepackage{hyperref}       
\usepackage[capitalise]{cleveref}
\usepackage{url}            
\usepackage{booktabs}       
\usepackage{amsfonts}       
\usepackage{nicefrac}       
\usepackage{microtype}      
\usepackage{xcolor}         

\title{Dynamic online matching with budget refills}

\author{
  Maria Cherifa \\
  Université de Paris, MAP5, Paris, France\\
  CREST, ENSAE, Palaiseau, France\\
  Faiplay joint team\\
  \texttt{maria.cherifa@ensae.fr} \\
   \And
   Clément Calauzènes \\
   Criteo AI Lab, Paris, France  \\
    Fairplay joint team  \\
   \texttt{c.calauzenes@criteo.com} \\
   \AND
   Vianney Perchet \\
  CREST, ENSAE, Palaiseau, France \\
  Criteo AI Lab, Paris, France\\
  Fairplay joint team\\
\texttt{vianney.perchet@normalesup.org} \\
}

\begin{document}

\maketitle

\begin{abstract}
 Inspired by sequential budgeted allocation problems, we explore the online matching problem with budget refills. Specifically, we consider an online bipartite graph $G=(U,V,E)$, where the nodes in $V$ are discovered sequentially and nodes in $U$ are known beforehand. Each $u\in U$  is endowed with a budget $b_{u,t}\in \lN$ that dynamically evolves over time. Unlike the canonical setting, where budgets are fixed, many practical applications involve periodic budget refills. This added dynamic introduces a richer and more complex problem, which we investigate here. Intuitively, adding extra budgets in $U$ seems to ease the matching task, and our results support this intuition. In fact, for the stochastic framework considered where we analyze the matching size built by $\greedy$  algorithm on an Erdős–Rényi random graph, we show that the matching size generated by $\greedy$ converges with high probability to a solution of an explicit system of ordinary differential equations (ODE). Moreover, under specific conditions, the competitive ratio (performance measure of the algorithm) can even tend to $1$. For the adversarial part,  where the graph considered is deterministic and the algorithm used is $\balance$, the $b$-matching bound holds when the refills are scarce. However, when refills are regular, our results suggest a potential improvement in the algorithm performance. In both cases, $\balance$ algorithm manages to reach the performance of the upper bound on the adversarial graphs considered.
\end{abstract}

\section{Introduction}

Finding matchings in bipartite graphs is a fundamental problem that lies at the intersection of graph theory \citep{Godsil1981MatchingsAW,Zdeborová_2006}, network theory, and combinatorial optimization \citep{lovasz2009matching,schrijver2003combinatorial}, with far-reaching implications across a wide range of practical applications, particularly in operations research, where it is often referred to as the assignment problem (see also \cite{Grove}). Specifically, a bipartite graph, denoted as \( G=(U,V,E) \), consists of two distinct sets of nodes, \( U \) and \( V \), and a set of edges \( E \subseteq U \times V \). These graphs provide a natural representation for systems in which entities from one set are paired with entities from the other, modeling relationships, dependencies, or allocations. Solving the matching problem involves identifying optimal pairings between nodes in the two sets while satisfying specific constraints.
\paragraph{The standard online matching problem.} Recent real-world applications, particularly in online advertising, have generated significant interest in the online variant of the matching problem (see \cite{mehta_survey}). In this version, the graph \( G = (U, V, E) \) consists of two sets of nodes $U
=\iset{n}:=\{1,\ldots,n\}$ and $V=\iset{T}:=\{1,\ldots,T\}$ for $n, T\in\mathbb{N}^*$ and set of edges \( E \subseteq U \times V \). The nodes in \( U \) are known beforehand, while the nodes in \( V \) arrive sequentially along with their associated edges. Each node \( u \in U \) has a budget \( b_{u,t} \geq 0 \), with an initial budget \( b_{u,0} = 1 \) for all \( u \in U \).

When a node \( t \in V \) is observed, an online algorithm must decide whether to match it with a node \( u \in U \) such that \( (u, t) \in E \) and \( u \) have not yet been matched (i.e. \( b_{u,t} = 1 \)). Once a matching decision is made, it cannot be changed, making the process irreversible. The budget of each node \( u \in U \) at time \( t \in V \) evolves according to the following dynamics: 

\begin{align}
    b_{u,t} = \begin{cases}
        b_{u,t-1} - 1 & \text{if \( u \) is matched to \( t \)}, \\
        b_{u,t-1} & \text{otherwise}.
    \end{cases}
\end{align}
Building on the foundational online matching problem, several generalizations have been proposed, such as the \( b \)-matching setting where each node in \( U \) has a budget \( b_u > 1 \), allowing multiple matches per node~\citep{Kalyan93,Khuller}. A node in \( V \) can only be matched if its neighbor in \( U \) has remaining budget~\citep{KALYANASUNDARAM2000319,b_matching_albers,Albers2022TightBF}. Algorithms like \( \greedy \) and \( \balance \) address these extensions, with \( \balance \) matching to the neighbor with the highest remaining budget. Their performance is typically evaluated via competitive ratios, comparing against the optimal offline matching~\citep{mehta_survey,feldman}.

\paragraph{The online matching with budget refills setting.} Motivated by the dynamic nature of online advertising, we model the scenario where $U$ represents the pool of campaigns or ads available to advertisers, and the nodes in $V$ correspond to the advertising slots that arrive sequentially. Each slot has varying eligibility for a subset of campaigns, determined by factors such as geographic location, browsing history, and other relevant features. The main objective of the advertiser is to maximize the number of ads displayed. In practice, campaigns are not displayed only once but come with a predetermined budget for impressions; for example, a particular ad may be shown no more than 10,000 times per day. This budget can evolve over time with the possibility of allocating additional resources to certain campaigns at periodic intervals, thereby motivating the model explored in this work.

Formally, we consider the same graph $G = (U, V, E)$ as in the standard matching problem. However, unlike the $b$-matching setting, an additional layer of complexity arises due to the dynamic nature of the budget for each node $u \in U$ at time $t \in V$. The budget, denoted by $b_{u,t} \in \mathbb{N}$, decreases whenever $u$ is matched but can also be replenished according to a replenishment process governed by $\eta_{u,t} \in \mathbb{N}$. In this work, we focus on a simple replenishment model where budgets are refilled at a constant average rate over time. This dynamic introduces significant challenges in analyzing budget evolution and its impact on matching performance.

\begin{align}
    b_{u,t} = \begin{cases}
        b_{u,t-1} - 1 + \eta_{u,t} & \text{if \( u \) is matched to \( t \)}, \\
        b_{u,t-1} + \eta_{u,t} & \text{otherwise}.
    \end{cases}
\end{align}

\paragraph{The budget refills.} 
As mentioned earlier, budget refills in online matching are inspired by online advertising, where advertisers replenish budgets to sustain exposure and revenue. We consider two refill settings. In the stochastic case, refills follow a Bernoulli distribution with parameter $\frac{\beta}{n}$, where $\beta > 0$ and $n$ is the number of nodes in $U$, reflecting real-world uncertainty and focusing on sparse graphs. This simple model highlights key dynamics without the complexity of more advanced random processes, which could obscure important insights. In the adversarial case, the graph is chosen to hinder matching, while budgets are refilled deterministically with one unit every $m$ steps, where $m$ depends on $n$. Both settings show similar behavior: stochastic refills occur at expected rate $\beta/n$, while adversarial refills happen at fixed rate $1/m$. Together, these models offer insight into how refill dynamics impact online matching, especially in advertising contexts.

\paragraph{The contributions.} Our main contributions, addressing both the stochastic and adversarial frameworks, are summarized as follows:
\begin{itemize}
    \item In the stochastic framework, we analyze the asymptotic performance of the $\greedy$ algorithm on the Erdős–Rényi model. Our approach demonstrates that, with high probability, the discrete and random matching process closely follows the continuous and deterministic solution of a system of ordinary differential equations. Further, we analyze the stability of the stationary solution of this system to derive a closed-form expression for the performance of $\greedy$. Finally, in terms of competitive ratio, we establish a lower bound on the competitive ratio which depends on different parameters of the problem. Notably, this lower bound converges to $1$ as these parameters approach infinity, indicating near-optimal performance in such limiting cases.
    
    \item In the adversarial case with relatively many refills (i.e., $m$ is negligible compared to $\sqrt{T}$). We show that the initial budgets have no significant impact on the asymptotic competitive ratio. We also derive an upper bound for the competitive ratio of 
$\balance$, demonstrating that it is optimal for the specific graph used in our proof. This bound, which applies to any algorithm, is given by $ 1 - \frac{1 - \alpha}{e^{1 - \alpha}}$, where $\alpha \simeq 0.603$.
    \item In the adversarial case, with relatively few refills, i.e., when $m$ is of order (larger than) $\sqrt{T}$, we demonstrate that refills have a negligible effect on the competitive ratio of the $\balance$ algorithm, which matches that of the $b_0$-matching problem (i.e., the problem without refills). Interestingly, the refill frequency $1/m$ does not appear in the competitive ratio. Stated otherwise, the dominating effect is the initialization of the budgets. 
\end{itemize}


\section{The adversarial framework}
\label{sec: adv_section}
Early research on online matching in the adversarial setting evaluated algorithms on the worst-case instances. The simple $\greedy$ algorithm guarantees a $1/2$ competitive ratio, improving to $1 - 1/e$ under random arrival \cite{mehta_goel}. The seminal $\ranking$ algorithm \citep{karp_vazirani} achieves the optimal worst-case ratio of $1 - 1/e$ and performs even better with random arrivals \citep{Devanur_ranking, ranking_made_simple, mahdian_yan}. Extensions to $b$-matching, where each node in $U$ has a budget $b$, were initiated in \cite{KALYANASUNDARAM2000319}, which introduced the $\balance$ algorithm with a competitive ratio approaching $1 - 1/e$ as $b$ increases. Later work extended this to non-uniform budgets $b_u$ and showed similar guarantees based on the minimum budget $b_{\min}$ see \citep{b_matching_albers}.

\subsection{Model}
To study the online matching problem in the adversarial setting with budget refills, it is essential to impose some restrictions on the adversary's power. Therefore, we adopt the following assumptions regarding the model being considered:
\begin{enumerate}
    \item The sequence of refills $(\eta_{u,t})_{u\in U, t\in V}$ is a parameter of the problem, set and known in advance to a refill of one unit every $m$ time steps.
    \item Every node $t\in V$ has at least one neighbor in $U$. 
\end{enumerate}
  If the sequence of refills $(\eta_{u,t})_{u\in U, t\in V}$ was to be chosen in an adversarial fashion, then the adversary would simply set it to 0, reducing the problem to the classical $b$-matching problem. Moreover to prevent thus reduction, we assume that each node $t\in V$ has at least one neighbor in $U$. Furthermore, the choice of a refill of one unit every $m$ time steps comes from the motivating application of advertising, where advertisers usually renew their budget monthly or quarterly.  Additionally, considering a constant value for refills gives a clear and simple setting to disentangle the asymptotic effect of the refills versus the initialization of budgets. 

Formally, the subset of graphs from which the oblivious adversary can choose is the following,
\begin{align}
    \cG_{T,m} = \left\{ (U, V, E, (\eta_{u,t})_{u \in U, t \in V}) : \forall t \in V,~ \exists u \in U \text{such that}~\eta_{u,t} = \indicator{t \bmod m = 0} \text{ and } (u,t) \in E \right\}
\end{align}

The budget evolution for each $u\in U$ is dependent on whether the edge $(u,t)\in E$ is included in the online matching, which is indicated by the binary variable $x_{u,t} \in \{0,1\}$, as well as whether $t$ is a multiple of $m$. Formally, this dynamic can be expressed as follows:
\begin{align}
   \forall u\in U, ~~ b_{u,t} = b_{u,t-1} - x_{u,t} +\indicator{t~{\rm mod}~m=0} ~~~~~~\text{and}~~~~~ b_{u,0} = b_0 ~~\text{for some}~ b_0 \geq 1.
\end{align}

As a consequence, the online \emph{matching} on $G \in \cG_{T,m}$ generated by an algorithm $\alg$ is the subset of edges that can be represented by a binary matrix $\mathbf{x} \in \{0,1\}^{n\times T}$ that must satisfy the following constraints:
\begin{enumerate}
    \item \label{item:first} $\forall (u,t)\in U\times V, ~(u,t)\not\in E \Rightarrow x_{u,t} = 0$ ~~(only edges in $E$ can be matched).
    \item \label{item:second} $\forall t\in V, ~ \sum_{u\in U} x_{u,t} \leq 1$ ~~($V$-nodes can only be matched once).
    \item \label{item:third}$\forall (u,t)\in U\times V, ~b_{u,t-1} < 1 \Rightarrow x_{u,t} = 0$ ($U$-nodes need  positive budget to be matched).
\end{enumerate}

In online bipartite matching problems, the performance of an algorithm $\alg$ is evaluated by its competitive ratio, which is the ratio between the size of the matching $\alg$ has created and the largest possible matching in hindsight, also referred to as $\opt$ with matrix $\bx^*$. The rationale is that the optimal matching of some deterministic graph $G$ can be arbitrarily small. Hence, the constructed matching size alone does not provide any good insight on the `quality'' of an algorithm in the adversarial case. Formally, in the adversarial framework, the objective of the algorithm is to obtain the highest worst-case competitive ratio $\CR^{\rm adv}(\alg, \cG_{T,m})$, defined as follows:

\begin{center}
\begin{minipage}{.3\textwidth}
    \begin{center}
    \end{center}
    \begin{align*}
        \CR^{\rm adv}(\alg, \cG_{T,m}) = \min_{G \in \cG_{T,m}} \frac{\alg(G)}{\opt(G)}
    \end{align*}
\end{minipage}
\end{center}
~\\
where $\alg(G)=\sum_{u\in U}\sum_{t=1}^T,x_{u,t}$, and $\opt(G)=\sum_{u\in U}\sum_{t=1}^T,x^*_{u,t}$, are the sizes of the matching generated by $\alg 
$ and $\opt$ respectively.

 As previously highlighted, our analysis focuses on evaluating the $\balance$ algorithm within the mentioned model. We aim to dissect the impact of the initial budget $b_0$ and the refill process by parameterizing our results with $T$, which is both the finite horizon and the size of $V$. This choice slightly limits the adversary's power, as it cannot impact the length of the horizon $T$ by simply providing no edge for an arbitrary number of time steps.
    

\subsection{Main results}
\subsubsection{ Regime $m=\omega(\sqrt{T})$}\label{ssec:big_m}
Intuitively, in the regime $m = \omega(\sqrt{T})$, the performance within a single period of length $m$ can have a dominant influence on the overall $\CR$. While this is clearly the case when $m=T$, it also holds true even for $m=\omega(T)$. 
\begin{restatable}{theorem}{ubdeterministicbmatching}\label{thm:ub_deterministic_12}
Assuming the initial budgets are $b_{1,0} = b_{2,0} = \dots = b_{n,0} = b_0 \geq 1$.   If $ m = \omega(\sqrt{T})$ and $b_0(b_0+1)^{b_0}\leq m$,  then, 
   \begin{center}
       \begin{equation}
       \sup_{{\tt ALG}:\text{deterministic}} \CR^{\rm adv}(\alg, \cG_{T,m}) \leq 1 - \frac{1}{\left(1+\frac{1}{b_0}\right)^{b_0}} + o_{T}\left(1\right)\,
   \end{equation}
   \end{center}
The bound is reached for the graph defined in the proof.
\end{restatable}

\begin{sketch}
    The complete proof is provided in \cref{app:proof_ub_deterministic_12}. It relies on using $d=\left\lfloor\frac{m}{|V_K|}\right\rfloor$ duplicates of the graph $G_K = (U_K, V_K, E_K)$ presented in \cite{KALYANASUNDARAM2000319}, where the size depends on $b_0$. More precisely, it uses $d$ copies of $G_K$ at the beginning of the process  and during the remaining time $T-m$ only one node $\tilde{u}$ from $U$ is connected with all the remaining edges in $V$ (see \cref{fig:graph_ub_adv2} for illustration). Then, the number of edges matched by $\alg$ and $\opt$ during these $T-m$ last steps is the same, denoted $\gamma_T$ which is at most $\left\lfloor\frac{T}{m}\right\rfloor$ as it relies on the refills of $\tilde{u}$ only (see \cref{app:proof_ub_deterministic_12} for more details).
Thus,

  $$  \CR^{\rm adv}(\alg, \cG_{T,m})  \leq \frac{d\alg(G_{K}) + \gamma_T}{d\opt(G_{K}) + \gamma_T}
        $$

Since $\gamma_T=o(\sqrt{T})$, we can conclude that, 
$$
\CR^{\rm adv}(\alg, \cG_{T,m})  \leq 1 - \frac{1}{\left(1+\frac{1}{b_0}\right)^{b_0}} + o_{T}\left(1\right)\,
$$

\end{sketch}
\subsubsection{Regime $m=o(\sqrt{T})$}\label{ssec:small_m}
In the regime $m = o(\sqrt{T})$, where the refills dominate the initialization, the upper bound on the $\CR$ is weaker.  Unlike the previous scenario where it was bounded by $1-\frac{1}{e} \approx 0.63$,  it is now bounded only by $0.73$. The following theorems establish this upper bound for $\balance$, and demonstrate that no algorithm can achieve significantly better performance.

\begin{restatable}{theorem}{ubbalance}\label{theo:ub_balance}
Assuming the initial budgets are $b_{1,0} = b_{2,0} = \dots = b_{n,0} = b_0 \geq 1$. For $m = o(\sqrt{T})$ and $mb_0 = o(T)$, then,
    \begin{align}
          \CR^{\rm adv}(\balance, \cG_{T,m}) \leq \underbrace{1- \frac{(1-\alpha)}{e^{(1-\alpha)}}}_{\simeq 0.73325...}+o_{m,T}(1)
    \end{align}
where $\alpha$ is defined by $\frac{1}{2}= \int_{0}^{\alpha} \frac{x e^x}{1-x}\mathrm{d}x$. The bound is reached for the graph defined in the proof.
\end{restatable}
\begin{sketch}
    The full proof is provided in \cref{app:proof_ub_balance}. It relies on building an adversarial graph $\GT=(U,V,E)$ with the following structure (see \cref{fig:graph_ub_adv1} for illustration): Initially, for a period of size $t_0 \simeq \frac{T}{e}$, the size of $U$ exceeds $m$ ($|U| \simeq 2m-1$), allowing the algorithm to accumulate a significant amount of budget. During this period, $\alg$ and $\opt$ consistently match nodes and accumulate an equal amount of budget on $U$, but it is not distributed in the same way. At time $t_0$, the adversary removes all but $m-1$ nodes from $U$ (starting with those with the highest budgets). Specifically, When the adversary eliminates nodes, it has no impact on $\opt$ because $\opt$ has perfect hindsight knowledge of the eliminated nodes. Therefore, it can allocate the budget exclusively to nodes that remain available, ensuring that no budget is lost on eliminated nodes at the time of their removal. However, $\alg$ remains unaware of which nodes will be eliminated. Consequently, at the time of elimination,  the nodes still have some budget.  Subsequently, the remaining nodes are removed one by one, at a rhythm that depends on  $m$, carefully designed for $\opt$ to widen the gap as much as
possible with $\alg$. 
\begin{figure}[h!]
    \centering
     \begin{subfigure}[t]{\textwidth}
     \label{fig:graph_ub_adv2}
        \centering
        \resizebox{0.7\textwidth}{!}{\input{figure_th_1}}
        \caption{The graph used for the proof of \cref{thm:ub_deterministic_12}}
    \end{subfigure}
     \vspace{1em}
    \begin{subfigure}[t]{\textwidth}
        \centering
        \resizebox{0.7\textwidth}{!}{\input{figure_th_2}}
        \caption{The graph $\GT$ used for the proof of \cref{theo:ub_balance}}
        \label{fig:graph_ub_adv1}
    \end{subfigure}
    \label{fig:combined_graphs}
\end{figure}

\paragraph{ALG vs OPT over time:}  As previously mentioned, up to time $t_{m-2}$, both $\alg$ and $\opt$ have the same performance. It is only between $t_{m-2}$ and $T$ that distinctions arise. Hence, the crucial step lies in determining the remaining budget of $\alg$ at time $t_{m-2}$ denoted $P_{t_{m-2}}$. To accomplish this, it's necessary to compute the values of $t_i$ and then analyze how the remaining budget of $\alg$ evolves over time. 

\paragraph{Intuition behind the choice of `$t_i$'' : } 

Since the main difference between $\alg$ and $\opt$ lies in the fact that $\opt$ knows the eliminated nodes beforehand, one important quantity to track is the 
$t_i$ which represents the time taken to deplete the budget of node $u_i$ by consistently avoiding matches with $u_i$ before $t_{i-1}$ and then matching it at every time step between $t_{i-1}$ and $t_i$ (a strategy employed by $\opt$). $t_i$ is determined by the following recurrence relationship: 
$$t_{i+1} \approx b_0 - 1 + t_i + \frac{b_0 + t_i}{m-1} $$
by solving it we get, 
$$t_i \approx  \left(1 + \frac{1}{m-1}\right)^i \left(t_0 + m b_0 - m + 1\right) - m b_0 + m - 1 $$

\paragraph{Intuition about the value of remaining budget:} 
The remaining budget at time $t_i$ follows the following recurrence,  
\begin{align*}
    P_{t_i} \approx \left(\underbrace{P_{t_{i-1}}}_{\text{the remaining budget at time}~t_{i-1}} + \underbrace{(n-i)(t_i - t_{i-1})/m}_{\text{the refills received between time $t_{i}$ and $t_{i-1}$ }} - \underbrace{(t_i - t_{i-1})}_{\text{number of nodes matched}}\right) \frac{n-i-1}{n-i}
\end{align*}
The expression $\frac{n-i-1}{n-i}$ represents the ratio of the number of nodes at time $t_{i}$ to the number of nodes at time $t_{i-1}$.

Therefore, the crux of the proof lies in examining the dynamics and rate of evolution of $P_{t_{i}}$ and handling the technicalities related to the approximations of the floor and ceil functions involved in the construction of the different quantities of the problem.

\end{sketch}
\subsubsection{No algorithm can beat  
 $\balance$}

\begin{restatable}{theorem}{uballadv}\label{theo:balance_is_best}
Assuming the initial budgets are $b_{1,0} = b_{2,0} = \dots = b_{n,0} = b_0 \geq 1$. For  $ m = o(\sqrt{T})$,
    \begin{align}
         \sup_{\alg}\lE\left[\CR^{\rm adv}(\alg, \cG_{T,m})\right] \leq   \CR^{\rm adv}(\balance, \GT)+o_{T}(1)
    \end{align}
    where the expectation is taken over the randomness from $\alg$.
\end{restatable}
The proof is provided in \cref{app:proof_balance_is_best}.

\begin{sketch}
    The intuition is that keeping budgets equalized between the currently available $U$-nodes is the best choice an algorithm can make in the adversarial graph used for the proof of \cref{theo:ub_balance}. This is because the adversary removes $U$-nodes one after the other, beginning with those with the highest budget and never providing again a $U$-node already removed.
    \end{sketch}

\medskip

\section{The Stochastic framework}
\label{sec: sto_section}
A significant line of research in online matching focuses on stochastic models, particularly the known i.i.d. model, where vertices arrive independently from a known distribution. This setting enables algorithms with improved competitive ratios—up to ~0.711—compared to $\ranking$ \citep{Manshadi_sto,jaillet_liu,Brubach2016Online,huang_shu}. However, such models are often idealized and fail to exploit richer structural information about the graph. As noted in \cite{brrodin}, simple greedy approaches can match or outperform these algorithms in average-case or practical scenarios, prompting interest in more realistic stochastic input models. Another research direction considers standard online algorithms applied to specific random graph classes. For instance, \cite{mastin_jaillet} study online matching in Erdős–Rényi graphs, especially in the sparse regime where edges exist independently with probability $c/n$. Even simple algorithms like $\greedy$ pose analytical challenges in this setting \citep{mastin_jaillet,second_member_approx,dyer_frieze}. More general models, such as the configuration model, which allows control over vertex degree distributions, have also been explored \citep{noiry2021online,aamand2022optimal}
\subsection{Model}
The online matching problem with refills of the budgets in the stochastic setting is studied in the following framework:
\begin{enumerate}
    \item The random graph is a standard  Erdős–Rényi model  $G(n,T,p)$, i.e., a bipartite graph with $n$ vertices on one side, $T$ on the other side and each potential edge $(u, t) \in U \times V$ occurs independently with probability $p$.
   \item The regime considered is the sparse one, in the sense that $p=\frac{a}{n}$ with $a>0$. This setup is motivated by online advertising, where the number of users greatly exceeds the number of ad campaigns, and only a small subset of users are eligible to participate. 
\item The sequence of refills $(\eta_{u,t})_{u\in U, t\in V}$ is a realization of independent Bernoulli random variable of parameter $\beta/n$, for some $\beta>0$.
\end{enumerate}
As emphasized previously, each node $u\in U$ is associated with a budget $b_{u,t}\in \lN$. We add the additional assumption that the maximum budget per node is capped at some $K\in \lN^{*}$ so that the budget dynamics are now expressed as follows, 
\begin{align}
    b_{u,t}=\min(K,b_{u,t-1}-x_{u,t}+ \eta_{u,t}) ~~~~~\text{with}~b_{u,0}=b_0\geq 1\, 
\end{align}
The reasons behind capping the maximal budget to $K$  are three-fold. First, in many applications in mind, the budget is capped (either by one, which corresponds to an idle/active state, or by a large number as in the online advertisement motivating example). Second, with the algorithm and the random graph considered, the budget will follow a negatively biased (and non-homogeneous) random walk, so that the maximal budget is sub-linear with arbitrarily high probability (hence this restriction is actually without loss of generality in the random model considered). Third, this capping induces a finite number of quantities to track through time (namely the current proportion of vertices with this or that budget), which greatly simplifies the analysis.

As defined previously,  an online \emph{matching} on $G$ generated by an algorithm $\alg$ is a subset of edges, represented by a binary matrix $\mathbf{x} \in \{0,1\}^{n\times T}$,  satisfying \Cref{item:first,item:second,item:third}.

 The performance of an algorithm in the stochastic setting can be either measured by the size of the expected matching size it creates or by the ratio between expected matching sizes of $\alg$ and $\opt$. Formally, the different quantities we shall consider are 
$$
        \CR^{\rm sto}(\alg,\cD) = \frac{\lE_{G\sim \cD}[\alg(G)]}{\lE_{G\sim \cD}[\opt(G)]}~~~~~~ \text{or}~~\text{matching size}=\lE_{G\sim \cD}[\alg(G)]
$$
where we denote by $\cD$ the distribution of graph considered and the refills,  $\alg(G)=\sum_{u\in U}\sum_{t=1}^T,x_{u,t}$, and $\opt(G)=\sum_{u\in U}\sum_{t=1}^T,x^*_{u,t}$ are the sizes of the matching generated by $\alg$ and $\opt$ respectively. Although the dependency on $T$ is implicit in $\alg(G)$ and $\opt(G)$, we will explicitly indicate it by using $\alg(G, T)$ and $\opt(G, T)$ instead.




\subsection{Main results}
Our first main theorem, stated below, identifies the asymptotic size of the matching generated by $\greedy$ on the bipartite Erdős-Rényi model with budget refills. The result shows that with high probability, the size of the matching generated by $\greedy$ is close to the solution of a system of ordinary differential equations.

\begin{restatable}{theorem}{lbstochasticgeneral}\label{theo:lb_stochastic_general}
With probability $1-\cO\left(n^{1/4} \exp(-a^3 n^{1/4})\right)$, the matching size created by $\greedy$ denoted by $\greedy(G,T)$ satisfies, 
$$\greedy(G,T)=nh(T/n)+\cO(n^{3/4})$$
and, 
$$\frac{\lE[\greedy(G,T)]}{n}\underset{n\to +\infty}{\to}h(T/n)$$
 where $h(\tau)$ is solution of the following equation,  
$$\dot{h}(\tau)=1-e^{-a(1-z_{0}(\tau))},~~~~~ \frac{1}{n}\leq \tau\leq \frac{T}{n}$$
and $z_0(\tau)$ satisfies the following system,
\begin{align}
    \begin{cases}
        \dot{z}_{0}(\tau)=-z_{0}(\tau)\beta + \frac{z_{1}(\tau)}{1-z_{0}(\tau)}(1-e^{-a+az_{0}(\tau)}) \hspace{1cm } & \text{for} ~k=0\\
        \dot{z}_{k}(\tau)= (z_{k-1}(\tau)-z_{k}(\tau))\beta +(z_{k+1}(\tau)-z_{k}(\tau))\frac{1-e^{-a+az_{0}(\tau)}}{1-z_{0}(\tau)} & \text{for} ~ 1\leq k \leq K-1\\
        \dot{z}_{k}(\tau)=\beta\,z_{k-1}(\tau)-z_{k}(\tau)\frac{1-e^{-a(1-z_0(\tau))}}{1-z_0(\tau)} &\text{for}~k= K\\
        \sum_{k=0}^K z_{k}(\tau)= 1
    \end{cases} \label{system_edo_z_k}
\end{align}

\end{restatable}

 \begin{sketch}
  For $0\leq k\leq K$,  $t\in \iset{T}$, let  $U_{k}(t)=\{u\in U: b_{u,t}=k\}$ be the set of nodes with budget equals to $k$ at time $t$ and $Y_{k}(t)=|U_{k}(t)|$  the total number of nodes with budget equals $k$ in $U$. The expectation of the one-step change of the variable $\greedy(G,t)$ can be expressed as, 
$$\mathbb{E}\left[\greedy(G,t+1)-\greedy(G,t)|\greedy(G,t)\right]= 1-\left(1-\frac{a}{n}\right)^{\sum_{k\geq 1}Y_{k}(t)}=1-\left(1-\frac{a}{n}\right)^{n-Y_{0}(t)} $$

As the evolution of $\greedy(G,t)$ depends on $Y_{0}$, an analysis of the process $\mathbf{Y}(t)=(Y_k(t))_{k\geq 0}$ is necessary. The dynamic of this process is described by the following system:
$$
\begin{cases}
        \mathbb{E}\left[\Delta_0(t)|\mathbf{Y}(t)\right]= -Y_{0}(t)\left[\frac{\beta}{n} (1-p\,\Sigma(t))\right]+Y_{1}(t)(1-\frac{\beta}{n})p\,\Sigma(t)\\
            \mathbb{E}\left[\Delta_1(t)|\mathbf{Y}(t)\right]= -Y_{1}(t)\left[\frac{\beta}{n} (1-p\,\Sigma(t))+ (1-\frac{\beta}{n})p\,\Sigma(t)\right]+Y_{0}(t)\frac{\beta}{n}+Y_{2}(t)(1-\frac{\beta}{n})p\,\Sigma(t)\\
            \mathbb{E}\left[\Delta_k(t)|\mathbf{Y}(t)\right]=\frac{\beta}{n} (1-p\,\Sigma(t))\left[Y_{k-1}(t)-Y_{k}(t)\right]
+\left[Y_{k+1}(t)-Y_{k}(t)\right](1-\frac{\beta}{n})p\,\Sigma(t)~~~~\forall k>1
        \end{cases}
$$

where $ \forall k\geq 0, ~\Delta_k(t)=Y_{k}(t+1)-Y_{k}(t)$, and $\Sigma(t)= \frac{1}{p(n-Y_{0}(t))} (1-(1-p)^{(n-Y_{0}(t))})$.  

 After establishing the evolution of these processes, the main idea behind the proof of \cref{theo:lb_stochastic_general} (postponed to \cref{app:proof_lb_stochastic_general}) is to show that $\greedy(G,T)$ is closely related to the solution of some ordinary differential equations (this is sometimes called the differential equation method     \citep{Wormald1995DifferentialEF, Warnke2019OnWD, Enriquez2019DepthFE} or stochastic approximations \citep{robbins_monro}. 
 \end{sketch}
\medskip
 After establishing that, with high probability $\greedy(G,T)$ is close to $nh(T/n)$, a function depending on $z_0(T/n)$, the solution of \cref{system_edo_z_k}, the objective is to solve this system to obtain an exact approximation of the matching size created by $\greedy$ on the Erdős–Rényi model. However, finding a closed-form solution of the system of differential equations \cref{system_edo_z_k} is quite challenging. To address this complexity, one approach is to explore the system's stationary solution and examine its stability; This means determining whether the solution to \cref{system_edo_z_k} converges to this stationary state, and then showing that $\greedy(G,T)$ converges towards a function depending on the stationary solution of \cref{system_edo_z_k}.

More precisely, \cref{corr:lb_stochastic_z_star_allk} shows  that for $K\geq 1$, $\greedy(G,T)$ converges with  high probability to $nh^{*}$ a function of $z_0^{*}$, the stationary solution of \cref{system_edo_z_k}. Additionally, as $n \to +\infty$, $\frac{\lE(\greedy(G, T))}{n}$ converges to $h^{*}(\psi)$. Furthermore, \cref{corr:lbstochastic_k=1} demonstrates at a specified rate the convergence of $\greedy(G,T)$ to $nh^{*}(T/n)$ with high probability and also, for $n\to +\infty$, $\frac{\lE(\greedy(G,T))}{n}$ converges to $h^{*}(T/n)$. The key distinction between these results lies in the nature of the convergence of $z_{0}(t)$ to $z_{0}^{*}$: in \cref{corr:lb_stochastic_z_star_allk}, $z_{0}(t)$ asymptotically converges to $z_{0}^{*}$, whereas in \cref{corr:lbstochastic_k=1}, the convergence is exponential.

\begin{restatable}{corollary}
{lbstochasticallk}\label{corr:lb_stochastic_z_star_allk} 
       For $K\geq 1$, with probability at least  $1-2\exp(-a^2 n^{\frac{3}{2}}/8T)$, 
$$\left|\greedy(G,T)-nh^*(T/n)\right|\leq o(T)$$
and, 
$$\frac{\lE[\greedy(G,T)]}{n}\underset{n\to +\infty}{\to}h^*(T/n)$$
with  $h^*(x)=\int_{1/n}^{x} (1-e^{-a(1-z_0^*)})\mathrm{d}\tau=\left(x-\frac{1}{n}\right)(1-e^{-a(1-z_0^*)})$, and  $z_0^*$ is the unique  solution of   $\sum_{k=0}^K z_0^*\left(\frac{\beta}{g(z_0^*)}\right)^k=1$ with $g(z_0^*)=\frac{1-e^{-a(1-z_0^*)}}{1-z_0^*}.$
\end{restatable}
Section \ref{app:proof_llb_stochastic_z_star_allk} contains the detailed proof.

\begin{sketch}
    The first step is to compute $\bar{S}_{z_0^{*}}$, the unique stationary solution of \cref{system_edo_z_k}. Then, to demonstrate that $\bar{S}_{z_0^{*}}$ is an asymptotically stable stationary solution, we rely on matrix perturbation theory. Once stability is established, we further prove that the matching size converges to a function that depends on $\bar{S}_{z_0^{*}}$.
    
\end{sketch}

\begin{restatable}{corollary}
{lbstochastickone}\label{corr:lbstochastic_k=1}
  For $K=1$, with probability at least  $1-2\exp(-a^2 n^{\frac{3}{2}}/8T)$, 
$$\left|\lE[\greedy(G,T)]-T(1-e^{-a(1-z_0^*)})\right|\leq c \frac{T}{(\log(T))^{3/4}}=o(T)$$
where  $z_0^*=\frac{1}{\beta}-\frac{1}{a}W\left(\frac{a}{\beta}e^{-a\left(1-\frac{1}{\beta}\right)}\right)$, with $W(\cdot)$ the Lambert function, and $c$ is some universal constant.
\end{restatable}

The proof is postponed to \cref{app:proof_llb_stochastic_k=1}.

\begin{sketch}
    For $K=1$,  \cref{system_edo_z_k} is reduced to a system of two equations. Firstly, we compute $S_{z_{0}^{*}}^{1}$, the stationary solution of the reduced system. Then, we prove that $S_{z_{0}^{*}}^{1}$ is an exponentially stable stationary solution using the perturbation method. Once the exponential stability is established, we further get that the matching size converges to a function depending only on $S_{z_{0}^{*}}^{1}$. 
\end{sketch}

The final main result of this section is summarized as follows. First, we establish a lower bound on $\CR^{\rm sto}$, which depends on $(z_0^*,\hdots,z_K^*)$ the stationary solution of \cref{system_edo_z_k}. This lower bound is derived through an exact calculation of the matching size achieved by the $\greedy$ algorithm and an upper bound on the matching size generated by $\opt$. Subsequently, we  demonstrate that  the competitive ratio converges to 1 as $T$, $K$ and $n$ grows significantly.

\begin{restatable}{proposition}
{lowerboundCRgeneralcase}\label{prop: lower_bound_CR_general_case}
For $T, K,n, b_0,\beta\in \lN^*$,
\begin{align}
     \CR^{\rm sto}(\greedy,\cD)\geq \frac{Tg(z_0^*)(1-z_0^*)+nb_0 -n\left(\frac{\beta}{g(z_0^*)-\beta}-\frac{(K+1)\beta^{K+1}}{g(z_0^*)^{K+1}-\beta^{K+1}}\right)} {nb_0 +\beta T } + \cO_{K,\beta}(T^{-1/4})
\end{align}
    where $\sum_{k=0}^K z_0^{*} \left(\frac{\beta}{g(z_0^*)}\right)^k=1$ with $g(z_0^*)=\frac{1-e^{-a(1-z_0^*)}}{1-z_0^*}$  as defined in \cref{corr:lb_stochastic_z_star_allk}.  
\end{restatable}

The proof is presented in \cref{app:proof_lower_bound_general_case}

\begin{sketch}
    Initially, we express $\greedy(G,T)$ as a function of $T, z_0(t),\beta,a$. Then, we use an upper bound on $\opt(G,T)$,  which is not very tight as it only takes into account the initial budget and the refills. Subsequently, we approximate $\greedy(G,T)$ by a function that depends on the stationary solution $\bar{S}_{z_0^{*}}$. It is noteworthy that in this context, the matching size $\greedy(G,T)$ aligns with that of \cref{theo:lb_stochastic_general} through the integration of the system \cref{system_edo_z_k}, up to negligible terms.
\end{sketch}

From \cref{prop: lower_bound_CR_general_case}, the next result  shows that when $K,T,n$ goes to infinity, the competitive ratio approaches 1. 

\begin{restatable}{theorem}{lbstochastickoneCR}\label{theo:lb_stochastic_k=1_CR}
For any $\alpha,\beta >0$, the competitive ratio tends to 1,  as $T,K,n$ approach infinity, as
    $$\underset{K,n \to +\infty}{\lim} ~\underset{T \to +\infty}{\lim} \CR^{\rm sto}(\greedy,\cD)=1$$
    \end{restatable}

The proof is in \cref{app:proof_lbstochastickoneCR}

\begin{sketch}
    The proof relies on calculating $z_0^{*}$ as $K$ approaches infinity.  Subsequently, as $T$ approaches infinity, the limit of $\CR^{\rm sto}(\greedy,\cD)$ is shown to be $g(z_0^{*})(1-z_0^{*})/\beta$. Finally, as $K$ tends towards infinity, the limit converges to $1$, with the assurance that this convergence happens with high probability as $n$ tends to infinity.
\end{sketch}

\section{Conclusion}
\label{conclusion}
We study online matching on a bipartite graph $G = (U, V, E)$ with dynamic budget refills on nodes in $U$, extending the classic framework. Two settings are considered: \emph{stochastic}, where refills follow a Bernoulli process in Erdős–Rényi graphs, and \emph{adversarial}, with deterministic graphs and refill patterns. In the stochastic case, $\greedy$ performs well under periodic refills, with its competitive ratio $\CR$ approaching 1. In the adversarial setting, infrequent refills have little impact on $\balance$, aligning its performance with $b$-matching. Frequent refills, however, yield better upper bounds on $\CR$. Establishing lower bounds is challenging due to dynamic budgets. A naive lower bound of $1 - \frac{1}{e}$ is obtained by duplicating nodes in a variant of $\ranking$. This leaves a gap with the upper bound $1 - \frac{1 - \alpha}{e^{1 - \alpha}}$ when $m = o(\sqrt{T})$.

\bibliographystyle{plainnat}
\bibliography{biblio}
\newpage
\appendix

\section{Adversarial Case}
\label{appendix:adv_case}

\subsection{Proof of \cref{thm:ub_deterministic_12}}\label{app:proof_ub_deterministic_12}
\ubdeterministicbmatching*

\begin{proof}
    Let $b_0, m, T \in \lN^*$ such that $m\geq kb_0$ where $k \triangleq (1+b_0)^{b_0}$ and $m \leq T$.
    The bipartite graph of size $(k, kb_0)$ used in \cite[Sec. 2, Thm. 5]{KALYANASUNDARAM2000319} is denoted $(U_0, V_0, E_0)$. To put the emphasis on which set of nodes the edges are defined on, $E_0$ will actually be denoted $E_0(U_0, V_0)$ as the structure of edges will be used on different subsets of nodes of the final graph.

    The graph $G = (U, \iset{T}, E)$ with $U = \{u_1, \dots, u_n\}$ of size $n\in \lN^*$ is built as follows. Intuitively, the first period of length $m$ is implementing copies of $E_0$ on disjoint nodes, then one remaining node in $U$ is the only neighbor of all following time steps. Denoting $j=\left\lfloor\frac{m}{k b_0}\right\rfloor$, 
    \begin{align}
        E &= \left(\bigcup_{i = 1}^{j} E_0\left(U_i, V_i\right)\right) \cup (\{\tilde{u}\} \times \irange{jkb_0+1}{T})
    \end{align}
    where $U_i = \{u_l : l\in \irange{(i-1)k+1}{ik}\}$,~~ $V_i=\irange{(i-1)mkb_0+1}{ikb_0}$ and $\tilde{u}$ is chosen to be a node of $U_1$ which has been depleted of its initial budget during $V_1$ (there is at least one).\\
    
For each $i \in\iset{j}$, on each subset $V_i$ of time steps, as per \citet[Proof of Thm. 5]{KALYANASUNDARAM2000319}, $\alg$ matches at most $b_0(b_0+1)^{b_0} - b_0^{b_0+1}$ edges, while $\opt$ manages to match $b_0(b_0+1)^{b_0}$ edges.
    After time $jkb_0$, both $\alg$ and $\opt$ match the same number of edges $\gamma_T$ which is at most the sum of refills obtained by $\tilde{u}$ -- i.e. $\left\lfloor\frac{T}{m}\right\rfloor$ -- (its initial budget is used during period $V_1$).
    In the end,
    \begin{align}
        \CR^{\rm adv}(\alg, \cG_{T,m}) & \leq \frac{j (kb_0  - b_0^{b_0+1}) + \gamma_T}{j k b_0 + \gamma_T}\\
        &\leq \frac{j (kb_0 - b_0^{b_0+1})}{j k b_0} + o_T(1) & \text{as}~\gamma_T=o(\sqrt{T})~\text{and}~jkb_0 = \omega(\sqrt{T})\\
        &=1 - \frac{1}{\left(1+\frac{1}{b_0}\right)^{b_0}} + o_T(1) & \text{def. of }k=(1+b_0)^{b_0}
    \end{align}

    Similarly, it is straightforward to show that $\balance$ achieves the lower bound of the $b$-matching problem on each of the duplicates of $E_0(U_i, V_i)$, as these sub-graphs are disjoint.
    
  \begin{figure}[h!]
   \centering
\input{figure_th_1}
    \caption{The graph used for the proof of \cref{thm:ub_deterministic_12}}
    \label{fig:extension_kalyanasundaram_appendix}
\end{figure}

\end{proof}

\subsection{Proof of \cref{theo:ub_balance}}\label{app:proof_ub_balance}
\ubbalance*

We provide a slightly more detailed result here.

\begin{theorem}\label{theo:ub_balance_complet}
Assuming the initial budgets are $b_{1,0} = b_{2,0} = \dots = b_{n,0} = b_0 \geq 1$. For $m = o(\sqrt{T})$ and $mb_0 = o(T)$, then,
    \begin{align}
         \CR^{\rm adv}(\balance, \cG_{T,m}) &\leq 1 - \frac{mb_0+t_0}{e(mb_0 + 2t_0)} - \frac{1}{e}\int_0^\alpha \frac{x^2e^x}{1-x}{\rm d}x \nonumber\\
         &+\frac{mb_0}{t_0}\left(1 - \frac{1}{e} + \frac{1}{e}\int_0^\alpha \frac{x(\alpha-x)e^x}{1-x}{\rm d}x\right) + o_{m,T}(1)
    \end{align}
 where $\alpha$ is defined as follows $\int_{0}^{\alpha} \frac{xe^x}{1-x}\mathrm{d}x = 1 - \frac{t_0}{mb_0+2t_0}$. The upper bound is reached for the graph defined in the proof.
\end{theorem}

The proof is organized as follows:
\begin{enumerate}
    \item Definition of the adversarial graph.
    \item Decomposition of $\balance(\GT)$.
    \item Several lemmas to treat each term of the decomposition.
\end{enumerate}

\paragraph{Definition of the adversarial graph for $\balance$.}

For $b_0, m, T \in \lN^*$ such that $m \leq T$, the number of $U$-nodes is set to $n = m-1 + \max\left(\left\lceil\frac{t_0}{b_0 + \left\lfloor\frac{t_0}{m}\right\rfloor}\right\rceil, \left\lceil\frac{m\left\lfloor\frac{t_0}{m}\right\rfloor}{b_0 + \left\lfloor\frac{t_0}{m}\right\rfloor-1}\right\rceil \right)$ (Note that when $b_0 \ll \frac{t_0}{m}$, then $n \simeq 2m - 1$). The graph $\GT=(U, V, E)$ is
defined as follows,
$$
 \begin{cases}
        U = \iset{n} \\
        V = \iset{T} \\
        E = (U\times \iset{t_0}) \cup (U_1 \times \irange{t_0+1}{t_1}) \cup \dots \cup (U_{m-1} \times \irange{t_{m-2}+1}{t_{m-1}} \cup (U_{m-1} \times \irange{t_{m-1}+1}{T})
    \end{cases}
    $$
where,
\begin{itemize}
    \item $U_1$ is the subset of the $m-1$ node with the lowest budget at time $t_0$, i.e.
$$U_1 \mathop\sim^{\text{unif}} \{A\subseteq U: |A| = m-1, \forall u \in A, u'\in U\setminus A, b_{u,t_0}\leq b_{u',t_0}\}\,.$$
    \item for any $i> 1$, $U_i$ is built be removing the node with the lowest budget at time $t_{i-1}$ from $U_{i-1}$ -- i.e. $U_i = U_{i-1} \setminus \{u_i\}$ where 
$$u_i \mathop\sim^{\text{unif}} \argmin_{u\in U_{i-1}}b_{u, t_{i-1}}\,.$$
    \item for any $i\geq 1$, $t_i = \inf\{t > t_{i-1}: b_0 + \left\lfloor\frac{t}{m}\right\rfloor = (t - t_{i-1})\}$. \textit{Intuition: $t_i$ is the time it takes to take the budget of $u_i$ to 0 by never matching $u_i$ before $t_{i-1}$ and matching it at every time step between $t_{i-1}$ and $t_i$ (which is what ${\tt OPT}$ does).}
    \item $t_0$ is chosen such that $T - t_{m-1} = o(T)$ (it is possible as proven in \cref{lem:T_t_m_1_=_oT})
\end{itemize}
\begin{figure}[h!]
   \centering
\input{figure_th_2}
    \caption{The graph $\GT$ used for the proof of \cref{theo:ub_balance} }
    \label{fig:graph_ub_adv_appendix}
\end{figure}
\begin{proof}
    The objective is to compute the performance of $\balance$ and $\opt$ on the graph $G^{\text{Th. 2}}$ defined above to obtain a bound on the $\CR$.
\paragraph{Performance of $\opt$.} Before time $t_0$, $\opt$ can use nodes from $U\setminus U_1$ to match a each time step: $|U\setminus U_1| = \max\left(\left\lceil\frac{t_0}{b_0 + \left\lfloor\frac{t_0}{m}\right\rfloor}\right\rceil, \left\lceil\frac{m\left\lfloor\frac{t_0}{m}\right\rfloor}{b_0 + \left\lfloor\frac{t_0}{m}\right\rfloor-1}\right\rceil \right)$, which ensures that the total budget of nodes in $U\setminus U_1$ over the period $\iset{t_0}$ is at least $t_0$ (accounting for the last refill that cannot necessarily be fully used). As a remark, if $b_0 = 1$, this simplifies to $|U\setminus U_1| = m$: with a refill every $m$ timesteps, $m$ nodes suffice to match at every time step. Thus, at time $t_0$, $\opt$ never matched any node from $U_1$. Then, by induction and definition of $t_i$, $\opt(\GT) = t_{m-1}$.

    \paragraph{Performance of $\alg$.}
    \begin{align}
        \balance(\GT) &= \underbrace{\sum_{t=1}^{t_0} \sum_{u\in U} x_{u,t}}_{\triangleq A_0}+ \sum_{i=1}^{m-1}\underbrace{\sum_{t=t_{i-1}+1}^{t_i}\sum_{u\in U_i} x_{u,t}}_{\triangleq A_i}\\
        &= A_0 + \sum_{i=1}^{m-1} A_i\\
        & = A_0 + \sum_{i=1}^{m-1} B^{(i)}_{t_{i-1}} - B^{(i)}_{t_{i}} + (m-i)\left(\left\lfloor\frac{t_i}{m}\right\rfloor - \left\lfloor\frac{t_{i-1}}{m}\right\rfloor\right) ~~~~~~~~~~~~~\text{(by induction)}\label{eq:18}\\
        \text{where } B^{(i)}_{t} = &\sum_{u\in U_i}b_{u,t}\,,\nonumber\\
        & = A_0 + \sum_{i=1}^{m-1} B^{(i)}_{t_{i-1}} - B^{(i)}_{t_{i}} + \sum_{i=1}^{m-1}(m-i)\left\lfloor\frac{t_i}{m}\right\rfloor - \sum_{i=0}^{m-2}(m-i-1) \left\lfloor\frac{t_{i}}{m}\right\rfloor\\
         & = A_0 + \left\lfloor\frac{t_{m-1}}{m}\right\rfloor - (m-1)\left\lfloor\frac{t_0}{m}\right\rfloor + \sum_{i=1}^{m-1} B^{(i)}_{t_{i-1}} - B^{(i)}_{t_{i}} +  \sum_{i=1}^{m-2} \left\lfloor\frac{t_{i}}{m}\right\rfloor\\
         & = A_0 - (m-1)\left\lfloor\frac{t_0}{m}\right\rfloor + \sum_{i=1}^{m-1} B^{(i)}_{t_{i-1}} - B^{(i)}_{t_{i}} + \sum_{i=1}^{m-1} \left\lfloor\frac{t_{i}}{m}\right\rfloor\\
         & = A_0 - (m-1)\left\lfloor\frac{t_0}{m}\right\rfloor + \sum_{i=1}^{m-1} B^{(i)}_{t_{i-1}} - B^{(i-1)}_{t_{i-1}} + B^{(i-1)}_{t_{i-1}} - B^{(i)}_{t_{i}} + \sum_{i=1}^{m-1} \left\lfloor\frac{t_{i}}{m}\right\rfloor\\
         & = A_0 - (m-1)\left\lfloor\frac{t_0}{m}\right\rfloor + B^{(0)}_{t_{0}} - B^{(m-1)}_{t_{m-1}} + \sum_{i=1}^{m-1} B^{(i)}_{t_{i-1}} - B^{(i-1)}_{t_{i-1}} + \sum_{i=1}^{m-1} \left\lfloor\frac{t_{i}}{m}\right\rfloor\\
         & = A_0 - (m-1)\left\lfloor\frac{t_0}{m}\right\rfloor + B^{(0)}_{t_{0}} - B^{(m-1)}_{t_{m-1}} + \sum_{i=0}^{m-2} B^{(i+1)}_{t_{i}} - B^{(i)}_{t_{i}} + \sum_{i=1}^{m-1} \left\lfloor\frac{t_{i}}{m}\right\rfloor\\
         & = A_0 - (m-1)\left\lfloor\frac{t_0}{m}\right\rfloor + B^{(1)}_{t_{0}} - B^{(m-1)}_{t_{m-1}} + \sum_{i=1}^{m-2} B^{(i+1)}_{t_{i}} - B^{(i)}_{t_{i}} + \sum_{i=1}^{m-1} \left\lfloor\frac{t_{i}}{m}\right\rfloor\\
         & = \underbrace{A_0 - (m-1)\left\lfloor\frac{t_0}{m}\right\rfloor}_{\triangleq Q_1} + \underbrace{B^{(1)}_{t_{0}}}_{\triangleq Q_2} - \underbrace{B^{(m-1)}_{t_{m-1}}}_{\triangleq Q_3} - \underbrace{\sum_{i=1}^{m-2} \left\lceil\frac{B^{(i)}_{t_{i}}}{m-i}\right\rceil}_{\triangleq Q_4}  + \underbrace{\sum_{i=1}^{m-1} \left\lfloor\frac{t_{i}}{m}\right\rfloor}_{\triangleq Q_5}
    \end{align}
    
    where the last equality comes from $B^{(i+1)}_{t_{i}} = B^{(i)}_{t_{i}} - \left\lceil\frac{B^{(i)}_{t_{i}}}{m-i}\right\rceil$ which in turn comes from the definition of $U_{i+1}$ (the adversary removes the node with most budget) combined with \cref{lemma:balance_keeps_equal_budgets} ($\balance$ equalizes budget among available nodes).
    
    The following lemma proves that $T=t_{m-1}+o(T)$, 
    \begin{lemma}
        \label{lem:T_t_m_1_=_oT}
        For $t_0\leq \frac{T}{e}$, $T=t_{m-1}+o(T)$.
    \end{lemma}
    
 \begin{proof}
According to \cref{lem:closed_forme_tilde_t_i}, 
\begin{align}
    \tilde{t}_{m-1} = \left(1 + \frac{1}{m-1}\right)^{m-1} \left(t_0 + m b_0 - m + 1\right) - m b_0 + m - 1
\end{align}
Putting everything together gives, 
\begin{align}
    T-t_{m-1}&=  T-t_{m-1}-\tilde{t}_{m-1}+ \tilde{t}_{m-1}\\
    &\leq T+ \left(1 + \frac{1}{m-1}\right)^{m-1} \left(t_0 + m b_0 - m + 1\right) - m b_0 + m - 1\\
    &\leq T+ e (t_0+mb_0)-mb_0
\end{align}
  choosing  $t_0$ such that $t_0= \lfloor T/e \rfloor $ along  with the fact that $m=o(\sqrt{T})$, implies that $T-t_{m-1}=o(T)$.
 \end{proof}

    \paragraph{Computation of the $\CR$.}
    
    \begin{align}
   & \CR^{\rm adv}(\balance, \cG_{T,m})\\
        &= \frac{Q_1 + Q_2 - Q_3 - Q_4 + Q_5}{t_{m-1} + o(T)} \\
        &= \frac{\cO\left(\frac{t_0}{m}\right) + Q_2 - Q_3 - Q_4 + Q_5}{t_{m-1} + o(T)}~~~~~~~~~~~~~~~~~~~~~~~~~~~~~~\text{as }A_0 = t_0 \\
        &= \frac{\cO\left(\frac{t_0}{m}\right) + (m-1)\left(b_0+\left\lfloor\frac{t_0}{m}\right\rfloor - \left\lfloor\frac{t_0}{n}\right\rfloor\right) + \cO(m) - Q_3 - Q_4 + Q_5}{t_{m-1} + o(T)}~~~~~~\text{\cref{lem:B_0_1}}\\
        &= \frac{\cO\left(\frac{t_0}{m}\right) + (m-1)\left(b_0+\left\lfloor\frac{t_0}{m}\right\rfloor - \left\lfloor\frac{t_0}{n}\right\rfloor\right) + \cO(m) - Q_4 + Q_5}{t_{m-1} + o(T)}\\
        & = \frac{1}{t_{m-1} + o(T)}\left(\cO\left(\frac{t_0}{m}\right) + (m-1)\left(b_0+\left\lfloor\frac{t_0}{m}\right\rfloor - \left\lfloor\frac{t_0}{n}\right\rfloor\right) - \lfloor \alpha^* m \rfloor \left\lceil\frac{B^{(1)}_{t_1}}{m}\right\rceil \right.\nonumber\\
        &~~~~~~ \left.+\bar{t}_0 \int_{\frac{1}{m}}^{\alpha^*} g_m(x){\rm d}x+ \frac{g_m(\alpha^*) - g_m(1/m)}{m} + \cO(m^2) + Q_5\right)~~~~~~~~~~~\text{\cref{lemma:sum_B_i_over_m}}\\
        & = \frac{1}{t_{m-1} + o(T)}\left(\cO\left(\frac{t_0}{m}\right) + (m-1)\left(b_0+\left\lfloor\frac{t_0}{m}\right\rfloor - \left\lfloor\frac{t_0}{n}\right\rfloor\right) - \lfloor \alpha^* m \rfloor \left\lceil\frac{B^{(1)}_{t_1}}{m}\right\rceil \right.\nonumber\\
        &~~~~~~ \left.+\bar{t}_0 \int_{\frac{1}{m}}^{\alpha^*} g_m(x){\rm d}x+ \frac{g_m(\alpha^*) - g_m(1/m)}{m} + \cO(m^2) \right.\nonumber\\
        &~~~~~~ \left.+ \left(\left(1 + \frac{1}{m-1}\right)^{m-1} - 1 + \frac{1}{m}\right) t_0 + \mathfrak{B}(m, b_0)\right)~~~~~~~~~~~~~~~~~~~~~~~~\text{\cref{lemma:sum_t_i_over_m}}
    \end{align}
    
    where $\alpha^* \in (1/m,1)$ the solution  of 
    \begin{align}
            \frac{\bar{t}_0}{m} \int_{\frac{1}{m}}^{\alpha^*} \frac{z}{1-z}e^z {\rm d}z - m\alpha^* - Y_1 = 0\,,
    \end{align}
    and 
    $$g_m(x) = \frac{x(\alpha^*-x)}{1-x}\left(1+\frac{1}{m-1}\right)^{mx}$$
    thus
    $\mathfrak{B}(m, b_0) \leq (e-2)mb_0 + b_0$
\end{proof}

\subsubsection{Proof of \cref{lemma:balance_keeps_equal_budgets}}
The following lemma states that $U$-nodes that were available exactly at the same time steps in the past should have the same budget within one unit when the algorithm is $\balance$.

\begin{lemma}\label{lemma:balance_keeps_equal_budgets}
    Let $W\subseteq U$ such that $\forall s\leq t \in V$, $(\exists u \in W, (u,s)\in E) \Rightarrow (\forall u \in W, (u,s)\in E)$. For the algorithm $\balance$,
    \begin{align}
        \exists \beta_t\in\lN,\forall u\in W, \exists z_{u,t}\in \{0,1\},~s.t.~ b_{u,t} = \beta_t + z_{u,t} ~~~\text{and}~~~ \sum_{u'\in W} z_{u',t} < |W|
    \end{align}
\end{lemma}
\begin{proof}
    We first focus on the first part of the result.
    Let $t \in \lN^*$ and $W\subseteq U$ such that $\forall s\leq t \in V$, $(\exists u \in W, (u,s)\in E) \Rightarrow (\forall u \in W, (u,s)\in E)$. We need to prove that using the $\balance$ algorithm implies that the budgets of the nodes at time $t$ differ only by one. We will prove it by recursion using the following hypothesis, 
    $$K(i): \exists \beta_i\in \mathbb{N}, \forall u\in W, b_{u,i}= \beta_i + z_{u,i}~~~~\text{with}~~~~z_{u,i}\in \{0,1\}~~~\text{and}~~~ \sum_{u'\in W} z_{u',i} < |W|$$
    
    By assumption, $\forall u\in W$ $b_{u,0}=b_0$, which means that $K(0)$ holds. 
    
    At time $i$, $\balance$ chooses $u_i \in \argmax_{u\in U: (u,i)\in E}b_{u,i}$.
    If $u_i\not\in W$, the result is direct from $K(i-1)$ with $\beta_i = \beta_{i-1} + \indicator{i ~{\rm mod}~m = 0}$. Otherwise, there are two cases when $u_i\in W$.
   
    \paragraph{Case $\forall u \in W, b_{u,i-1}=\beta_{i-1}$.} Then, by choosing $\beta_i = \beta_{i-1} - 1 + \indicator{i ~{\rm mod}~m = 0}$, we have $b_{u_i, i} = \beta_i$ and for any $u'\in W\setminus\{u_i\}, b_{u', i} = b_{u, i-1} + \indicator{i ~{\rm mod}~m = 0}$.

    \paragraph{Case $\exists u,u' \in W, b_{u,i-1}\neq b_{u', i-1}$.} Then, by choosing $\beta_i = \beta_{i-1} + \indicator{i ~{\rm mod}~m = 0}$, we have $b_{u_i, i} = \beta_i$ and $\forall u\in W\setminus\{u_i\}, b_{u,i} = b_{u, i-1} + \indicator{i ~{\rm mod}~m = 0}$.

    In both cases, $K(i)$ holds.\\
\end{proof}

\subsubsection{Proof of \cref{lem:B_0_1}}

During the phase $i$, between $t_{i-1}$ and $t_i$, the graph is fully-connected to $U_i$. Thus, $\sum_{u\in U_i} b_{u,t}$ follows the following dynamic $Z_t$.

Given $k, m, t, j\in \mathbb{N}^*$, the dynamic of interest is
\begin{align}
    Z_t = Z_{t-1} - \indicator{Z_{t-1} \geq 1} + k \indicator{t~{\rm mod}~m=j}\,.
\end{align}
where $k = |U_i|$, and $j$ accounts for the fact that a phase begins at a time $t_{i-1}$ that is not necessarily a multiple of $m$.

\begin{lemma}
\label{Z_dynamic}
    For $k,m, t, Z_0, j \in \lN^*$, 
    \begin{align}
        Z_t = 
        \begin{cases}
            (Z_0 - t)_+ + k\indicator{t=j} & \text{if } t \leq j\\
            g(Z_j, k, t - j, m) & \text{if } j < t \leq j + t^*\\
            f(k,m,t-j, \tilde{t}) & \text{if } j + \tilde{t} \leq t 
        \end{cases}
    \end{align}
    \begin{align}
        t^* =
        \begin{cases}
            Z_j + k\left\lceil\frac{Z_j+1-m}{m-k}\right\rceil & \text{if } m>k\\
            Z_j & \text{if }m\leq k \text{ and } Z_j < m\\
            +\infty & \text{otherwise}
        \end{cases}
        & &\text{ and } \tilde{t} = m\left\lceil\frac{t^*}{m}\right\rceil
    \end{align}
    and $f(k,m,t,\Tilde{t})=\left(
        \indicator{k < m}(k-(t~{\rm mod}~m ))_{+}
         +
        \indicator{k \geq m}\left(k \left(1+\left\lfloor \frac{t-\tilde{t}}{m}\right\rfloor\right) - (t-\tilde{t})\right)
         \right)$ and  $~~g(Z_j,k,t,m)=\left(Z_j + k\left\lfloor \frac{t}{m}\right\rfloor - t \right)$
\end{lemma}
\begin{proof}
\textbf{First, the case when $\mathbf{j=0}$.}\\
    For $k, m, t, Z_0 \in \mathbb{N}^*$, $t^*$ is defined to be the first time at which $Z_t$ reaches $0$.
    \begin{align}
        t^* & = \min_{t\in\lN^*} t ~~~~\text{s.t. }Z_t=0
    \end{align}
    which value is given by \cref{lem:opt_t_star}.
    
    \textbf{For any $t \leq t^*$}, $\indicator{[Z_{t-1} > 0]} = 1$, thus, by recursion,
    \begin{align}
        Z_t &= Z_0 - t +  k\sum_{t'=1}^t\indicator{[t' ~{\rm mod}~ m =0]}\\
        & = Z_0 - t + k\left\lfloor\frac{t}{m}\right\rfloor\,.
    \end{align}

    \textbf{For any $t^* \leq t < \tilde{t}$, $t~{\rm mod}~m \neq 0$} and thus $Z_t = 0$ (by recursion starting at $Z_{t^*}=0$).\\

    \textbf{For any $t > \tilde{t}$}, the analysis is split between the case $k\geq m$ and $k < m$. In both cases $Z_{\tilde{t}} = k$ and we denote $t = \tilde{t} + \Delta t$.
 
    First, if $k\geq m$, similarly as before, we get $Z_{\tilde{t} + \Delta t} = (Z_{\tilde{t}} - \Delta t + k\left\lfloor\frac{\Delta t}{m}\right\rfloor)$: for $k\geq m$, it is always true that $Z_{\tilde{t} + \Delta t-1}>0$ which gives the result by recursion.\\

    Second, if $k < m$, the result is proved by recursion.
    
    Recursion hypothesis -- $H(t) = Z_{\tilde{t}+\Delta t} = \left(k - (t~{\rm mod}~m)\right)_+$ and by definition of $\tilde{t}$, $H(\tilde{t})$ holds as $Z_{\tilde{t}} = k$.
    \begin{align}
        Z_{t + 1} 
        & = Z_{t} -\indicator{Z_t > 0} + k\indicator{t + 1 ~{\rm mod}~m = 0}\\
        & = (k - (\tilde{t}+\Delta t)~{\rm mod}~m)_+ -\indicator{(k - (\tilde{t}+\Delta t)~{\rm mod}~m))_+ > 0} + k\indicator{(\tilde{t}+\Delta t+1) ~{\rm mod}~m = 0}\\
        & = (k - \Delta t~{\rm mod}~m)_+ -\indicator{(k - \Delta t~{\rm mod}~m) > 0} + k\indicator{(\Delta t+1) ~{\rm mod}~m = 0}
    \end{align}
    \begin{itemize}
        \item \textbf{If $(\Delta t+1)~{\rm mod}~m = 0$,}  we necessarily have $Z_{\tilde{t}+\Delta t} = 0$ (as $k \leq m-1$). Thus $Z_{\tilde{t}+\Delta t + 1} = k = (k - (\tilde{t}+\Delta t + 1)~{\rm mod}~m)_+$.
        \item \textbf{If $(\Delta t+1)~{\rm mod}~m \neq 0$ and $Z_{\tilde{t}+\Delta t} >0$,} we have $Z_{\tilde{t}+\Delta t + 1} = Z_{\tilde{t}+\Delta t} - 1$ which gives the result,
        \item \textbf{If $(\Delta t+1)~{\rm mod}~m \neq 0$ and $Z_{\tilde{t}+\Delta t} = 0$,} we have $Z_{\tilde{t}+\Delta t + 1} = Z_{\tilde{t}+\Delta t}$ which gives the result.
    \end{itemize} 
\textbf{Second, the general case when $\mathbf{0\leq j < m}$.}\\
Let $\tilde{Z}_t = Z_{t+j}$, $t_j^* = \min_{t\in\lN^*} t$ s.t. $\tilde{Z}_t=0$ and $\tilde{t}_j = m\left\lfloor\frac{t_j^*}{m}\right\rfloor$.
Using the result proved above for $j=0$ gives for any $t > j$ 
\begin{align}
    \tilde{Z}_{t - j} &= g(\tilde{Z}_0, k, t - j, m)\indicator{t-j \leq t_j^* } +\indicator{t-j \geq \tilde{t}_j} f(k,m,t-j, \tilde{t})\\
    & \Leftrightarrow Z_t = g(Z_j, k, t - j, m)\indicator{t-j \leq t^* } +\indicator{t-j \geq \tilde{t}} f(k,m,t-j, \tilde{t})
\end{align}
and for any $t \leq j$,
\begin{align}
    Z_t = (Z_0 - t)_+ + k\indicator{t = j}
\end{align}
\end{proof}

\begin{lemma}\label{lem:B_0_1}
    $B_{t_0}^{(1)} = (m-1)\left(b_0+\left\lfloor\frac{t_0}{m}\right\rfloor - \left\lfloor\frac{t_0}{n}\right\rfloor\right) + (m-1-(t_0 ~{\rm mod}~n)_+)$
\end{lemma}
\begin{proof}
    By application of \cref{Z_dynamic}, $B_{t_0}^{(0)} = n b_0 + n\left\lfloor\frac{t_0}{m}\right\rfloor - t_0$. By application of \cref{lemma:balance_keeps_equal_budgets}, 
    \begin{align}
        B_{t_0}^{(1)} 
        & = (m-1)\left\lfloor\frac{n b_0 + n\left\lfloor\frac{t_0}{m}\right\rfloor - t_0}{n}\right\rfloor + (m-1-(t_0 ~{\rm mod}~n))_+\\
        & = (m-1)\left(b_0+\left\lfloor\frac{t_0}{m}\right\rfloor - \left\lfloor\frac{t_0}{n}\right\rfloor\right) + (m-1-(t_0 ~{\rm mod}~n))_+
    \end{align}
\end{proof}

\subsubsection{Proof of \cref{lemma:sum_t_i_over_m}}

This section is organized as follows: 
\begin{enumerate}
    \item A characterization of $t_i$ by a recursive equation.
    \item The introduction of $\tilde{t}_i$ (to approximate $t_i$).
    \item The quantification of the approximation error between $t_i$ and $\tilde{t}_i$. 
    \item A closed-form computation of $\sum_{i=1}^{m-1}\tilde{t_i}$.
    \item The final result.
\end{enumerate}

\paragraph{A characterization of $t_i$.} The following result allows to characterize the sequence of $t_i$ by a recursive equation.

\begin{lemma}\label{lem:opt_t_star}
    For $a, b \geq 0, m, c \geq 2$, if $t^* = \inf\{t\in\lN^*: b + c\left\lfloor\frac{t}{m}\right\rfloor = (t - a)\}$ then,
    \begin{align}
        t^* = \begin{cases}
            a+b + c\left\lceil\frac{a+b+1-m}{m-c}\right\rceil & \text{if } m>c\\
            a+b & \text{if }m\leq c \text{ and } a+b < m\\
            +\infty & \text{otherwise}
        \end{cases}
    \end{align}
\end{lemma}
\begin{proof}
First,
    \begin{align}
        t^* & = \min_{t\in \lN^*} t ~~~\text{s.t.}~~ b + c\left\lfloor\frac{t}{m}\right\rfloor - (t - a) = 0\\
        & = \min_{t\in \lN^*} t ~~~\text{s.t.}~~ b + a + (c-m)\left\lfloor\frac{t}{m}\right\rfloor - \left\{\frac{t}{m}\right\} = 0 ~~~~~ \left(\{\cdot\} \text{ denotes the fractional part}\right)\label{eq:prob_opt_t_i:2}\\
        & = \min_{\substack{k\in\lN\\j\in[0 \dots m-1]}} km+j ~~~~ \text{s.t. } (a+b) - j + (c-m)k = 0 \text{ and }km+j>0\label{eq:prob_opt_t_i:3}\\       
        & = \min_{k\in\lN} a + b + c k ~~~~ \text{s.t. } (a+b) +1 -m \leq (m-c)k \leq a + b \text{ and }a+b+k > 0 \label{eq:prob_opt_t_i:4}\\       
    \end{align}
    Going from \cref{eq:prob_opt_t_i:2} to \cref{eq:prob_opt_t_i:3} is done by using the Euclidean division of $t$ by $m$ as $t = km + j$.
    As \cref{eq:prob_opt_t_i:4} is linear in $k$ with positive coefficients, it is minimized at the lowest feasible value of $k$ which is
    \begin{align}
        k^* = \begin{cases}
            \left\lceil\frac{a+b+1-m}{m-c}\right\rceil & \text{if } m>c\\
            0 & \text{if }m\leq c \text{ and } a+b < m\\
            +\infty & \text{otherwise}
        \end{cases}
    \end{align}
    The result follows by using the fact that $t^* = k^*m + j^*$ where $j^* = a+b+(1-m)k^*$. 
\end{proof}

\begin{corollary}\label{cor:charac_t_i}
$\forall i\in\lN, t_{i+1} = b_0 - 1 + t_i + \left\lceil\frac{b_0 + t_i}{m-1}\right\rceil\,.$
\end{corollary}
\begin{proof}
    Direct application of \cref{lem:opt_t_star} from the definition of $t_i$.
\end{proof}

\paragraph{Introduction of $\tilde{t}_i$ to approximate $t_i$.}
To obtain a sequence $\tilde{t}_i$ close to $t_i$ with a closed form, the intuition is to "remove" the fractional part and solve the arithmetic-geometric equation.
\begin{align}\label{eq:recursion_tilde_t_i}
    &\forall i\in\lN^*, \tilde{t}_{i+1} = b_0 - 1 + \tilde{t}_i + \frac{b_0 + \tilde{t}_i}{m-1},.\\
    &\tilde{t}_0 = t_0
\end{align}
\textit{The intuitive justification is that we are in the regime $m=o(\sqrt{T})$, thus the error introduced by ignoring a term of order $\left\{\frac{b_0+t_i}{m-1}\right\}$ is small (especially if $t_1 = \Theta(T)$).}\\
Now, given that $\tilde{t}_i$ follows an arithmetic-geometric equation, it admits a closed-form expression:
\begin{lemma}\label{lem:closed_forme_tilde_t_i}
For any $i\in\lN$,
\begin{align}
\label{eq:tilde_t_i}
    \tilde{t}_i = \left(1 + \frac{1}{m-1}\right)^i \left(t_0 + m b_0 - m + 1\right) - m b_0 + m - 1
\end{align}
\end{lemma}
\begin{proof}
    Let $i$ be in $\lN$.
\begin{align}
        \tilde{t}_{i+1} &= b_0 - 1 + \tilde{t}_i + \frac{b_0 + \tilde{t}_i}{m-1} \\
        & \Leftrightarrow \tilde{t}_{i+1} = \frac{m}{m-1}\tilde{t}_i + \frac{m}{m-1}b_0 - 1\\
        & \Leftrightarrow \tilde{t}_{i+1} + mb_0 - m + 1 = \frac{m}{m-1}\tilde{t}_i + \frac{m}{m-1}b_0 - 1 + mb_0 - m + 1\\
        & \Leftrightarrow \tilde{t}_{i+1} + mb_0 - m + 1 = \frac{m}{m-1} \left(\tilde{t}_i + b_0 + (m-1)b_0 - m + 1\right)\\
        & \Leftrightarrow \tilde{t}_{i+1} + mb_0 - m + 1 = \frac{m}{m-1} \left(\tilde{t}_i + m b_0 - m + 1\right)\\
        & \Leftrightarrow \tilde{t}_{i+1} = - mb_0 + m - 1 + \left(\frac{m}{m-1}\right)^i (\tilde{t}_0 + m b_0 - m + 1)
    \end{align}
\end{proof}

\paragraph{Quantification of the approximation error.}
\begin{lemma}\label{lem:error_t_tilde_t}
    $\forall i \in \lN^*, t_i - \tilde{t}_i < (m-1)\left(\left(1+\frac{1}{m-1}\right)^i -1\right)$. 
\end{lemma}
\begin{proof}
\begin{align}
        t_i - \tilde{t}_i & = t_{i-1} - \tilde{t}_{i-1} + \left\lceil\frac{t_{i-1} + b_0}{m-1}\right\rceil - \frac{\tilde{t}_{i-1} + b_0}{m-1} ~~~~~~~~~~~~~~~~~\text{\cref{eq:recursion_tilde_t_i,cor:charac_t_i}} \\
        & = t_{i-1} - \tilde{t}_{i-1} + \frac{t_{i-1} - \tilde{t}_{i-1}}{m-1} +\left\lceil\frac{t_{i-1} + b_0}{m-1}\right\rceil - \frac{t_{i-1} + b_0}{m-1}  \\
        & = (t_{i-1} - \tilde{t}_{i-1}) \left(1 + \frac{1}{m-1}\right) + \left\lceil\frac{t_{i-1} + b_0}{m-1}\right\rceil - \frac{t_{i-1} + b_0}{m-1} 
    \end{align}
    Thus, by induction, using that $t_0 - \tilde{t}_0 = 0$ (by definition).
    \begin{align}
\label{eq:error_t_tilde_i_intermediary}
        t_i - \tilde{t}_i & < (t_{i-1} - \tilde{t}_{i-1}) \left(1 + \frac{1}{m-1}\right) + 1 \\
\label{eq:error_t_tilde_i_intermediary_bis}
& < \left(1 + \frac{1}{m-1}\right)^i(m-1) + 1 - m
    \end{align}
\end{proof}

\paragraph{Closed-form computation of $\sum_{i=1}^{m-1}\tilde{t}_i$.}

\begin{lemma}\label{lem:sum_tilde_t}
    $\sum_{i=0}^{m-1} \tilde{t}_i = \left(\left(1 + \frac{1}{m-1}\right)^{m-1} - 1 + \frac{1}{m}\right) m t_0 + \mathfrak{A}(m, b_0)$ where 
    \begin{align}
        \mathfrak{A}(m, b_0) = m (mb_0 -m +1) \left(\left(1 + \frac{1}{m-1}\right)^{m-1} - 2 + \frac{1}{m}\right) \,.
    \end{align}
\end{lemma}
\begin{proof}
    \begin{align}
        \sum_{i=0}^{m-1} \tilde{t}_i 
        & = -m(mb_0 - m +1) + \left(t_0 + m b_0 -m + 1\right)\sum_{i=0}^{m-1} \left(1 + \frac{1}{m-1}\right)^i\\
        & \\
        & = -m(mb_0 - m +1) - m + \left(t_0 + m b_0 -m + 1\right)\frac{\left(1 + \frac{1}{m-1}\right)^m -1}{1 + \frac{1}{m-1} - 1} \\
        & = -m(mb_0 - m +1) - m + \left(t_0 + m b_0 -m + 1\right)(m-1)\left(\left(1 + \frac{1}{m-1}\right)^m -1\right) \\
        & = \left(\left(1 + \frac{1}{m-1}\right)^{m-1} - 1 + \frac{1}{m}\right) m t_0 \nonumber\\
        & ~~~~~~~~~~~~~~~~~~~~~~~~~~~~~~~~~~+
        \underbrace{m (mb_0 -m +1) \left(\left(1 + \frac{1}{m-1}\right)^{m-1} - 2 + \frac{1}{m}\right)}_{\triangleq \mathfrak{A}(m, b_0)}
    \end{align}
\end{proof}

\paragraph{Putting everything together.}

\begin{lemma}\label{lemma:sum_t_i_over_m}
    $\sum_{i=1}^{m-1} \left\lfloor\frac{t_{i}}{m}\right\rfloor = \left(\left(1 + \frac{1}{m-1}\right)^{m-1} - 1 + \frac{1}{m}\right) t_0 + \mathfrak{B}(m, b_0)$ where $\mathfrak{B}(m, b_0) < (e-2)m b_0 + b_0$
\end{lemma}
\begin{proof}
    \begin{align}
        \sum_{i=1}^{m-1} \left\lfloor\frac{t_{i}}{m}\right\rfloor 
        & = \sum_{i=1}^{m-1} \frac{\tilde{t}_{i}}{m}  + \underbrace{\sum_{i=1}^{m-1} \frac{t_i - \tilde{t}_{i}}{m} - \sum_{i=1}^{m-1} \left\{\frac{\tilde{t}_{i}}{m}\right\}}_{\triangleq \mathfrak{E}(m)} 
    \end{align}
    where
    \begin{align}
        \mathfrak{E}(m)& \leq  \sum_{i=1}^{m-1} \frac{t_i - \tilde{t}_{i}}{m} \\
        & <  \frac{m-1}{m}\sum_{i=1}^{m-1}\left(\left(1+\frac{1}{m-1}\right)^i -1\right) \\
        & < \frac{m-1}{m}\left(\frac{\left(1+\frac{1}{m-1}\right)^m - 1}{1+\frac{1}{m-1} - 1} - m\right)  \\
        & < \frac{(m-1)}{m}\left((m-1)\left(\left(1+\frac{1}{m-1}\right)^m - 1\right) - m\right)  \\
        & < \frac{(m-1)}{m}\left(m\left(1+\frac{1}{m-1}\right)^{m-1} + 1 - 2m\right) \\
        & < {(m-1)\left(\left(1+\frac{1}{m-1}\right)^{m-1} - 2\right) + \frac{m-1}{m}} 
    \end{align}
    Using \cref{lem:sum_tilde_t} gives
    \begin{align}
        \sum_{i=1}^{m-1} \left\lfloor\frac{t_{i}}{m}\right\rfloor 
        & = \left(\left(1 + \frac{1}{m-1}\right)^{m-1} - 1 + \frac{1}{m}\right) t_0 + \underbrace{\mathfrak{E}(m) + \mathfrak{A(m,b_0)}}_{\triangleq \mathfrak{B}(m, b_0)}
    \end{align}
    where
    \begin{align}
        \mathfrak{B}(m, b_0)& <\left(\left(1 + \frac{1}{m-1}\right)^{m-1} - 2 + \frac{1}{m}\right)mb_0
    \end{align}
\end{proof}

\subsubsection{Proof of \cref{lemma:sum_B_i_over_m}}
The objective of this subsection is the compute $\sum_{i=1}^{m-1} \left\lceil\frac{B^{(i)}_{t_i}}{m-i}\right\rceil$. This section is organised as follows: 
\begin{enumerate}
    \item The introduction of $Y_i$ to approximate.
    \item The bounding of $Y_i$.
    \item The quantification of the approximation error between $Y_i$ and $\left\lceil\frac{B^{(i)}_{t_i}}{m-i}\right\rceil$. 
    \item The final result.
\end{enumerate}
\paragraph{Introduction of $Y_i$.} For any $i\geq 1$, $Y_i$ is defined by the following recursion:
\begin{align}\label{eq:def_Y_i_adv}
    Y_1 &= \left\lceil\frac{B^{(1)}_{t_1}}{m-1}\right\rceil \\
    Y_{i+1} & = Y_i - (\tilde{t}_{i+1} - \tilde{t}_i)\left(\frac{1}{m-i-1}-\frac{1}{m}\right) + 1
\end{align}
where $\tilde{t}_i$ is the approximate time dynamic defined in \cref{eq:recursion_tilde_t_i}.

\paragraph{The bounding of $Y_i$}.
\begin{lemma}\label{lem:y_i_dynamic}
    For $1\leq i < m-1$,
    \begin{align}
    Y_1 +i - 1 - \frac{\bar{t}_0}{m} g((i+1)/m) \leq Y_i \leq Y_1 +i - 1 - \frac{\bar{t}_0}{m} g(i/m)
    \end{align}
    where $g(z) = \int_{\frac{1}{m}}^{z}\frac{x}{1-x}\exp(x)dx$ and $\bar{t}_0 = t_0 - mb_0 -m + 1$.
\end{lemma}
\begin{proof}
By definition of $Y_i$, it holds
        \begin{align}
        Y_i & = Y_1 + i-1 - \sum_{k=2}^i (\tilde{t_k} - \tilde{t}_{k-1})\left(\frac{1}{m-k} - \frac{1}{m}\right)\\
        & = Y_1 + i-1 - \sum_{k=2}^i (\tilde{t_k} - \tilde{t}_{k-1})\frac{k}{m(m-k)} \\
        & = Y_1 + i-1 - \frac{\bar{t}_0}{m}\sum_{k=2}^i \left(1 + \frac{1}{m-1}\right)^k \left(1 - \frac{m-1}{m}\right)\frac{k}{(m-k)} & \text{by }\cref{eq:tilde_t_i} \\
        & = Y_1 + i-1 - \frac{\bar{t}_0}{m^2}\sum_{k=2}^i \frac{k}{m-k} \left(1 + \frac{1}{m-1}\right)^k\label{eq:y_i_exact} 
    \end{align}
    where $\bar{t}_0 = t_0 - mb_0 -m + 1$.
Moreover, since $(1+\frac{1}{m-1}) \geq \exp(\frac{1}{m})$, this gives 
$$\exp(\frac{k}{m})\leq ( 1+\frac{1}{m-1})^k\leq \exp(\frac{k}{m-1})\leq \exp(\frac{k}{m})(1+\frac{2}{m}),$$
Since the function $x \mapsto \frac{x}{1-x}\exp(x)$ is increasing on $\mathbb{R}_+$, we get that (for $i<m-1$) 
$$ \underbrace{\int_{\frac{1}{m}}^{\frac{i}{m}}\frac{x}{1-x}\exp(x)dx}_{\triangleq g(i/m)} \leq  \underbrace{\frac{1}{m}\sum_{k=2}^i\frac{k}{m-k}\exp(\frac{k}{m})}_{A}\leq \int_{\frac{2}{m}}^{\frac{i+1}{m}}\frac{x}{1-x}\exp(x)dx,$$
Or equivalently, for $i < m-1$,
\begin{align}
    Y_1 +i - 1 - \frac{\bar{t}_0}{m} g((i+1)/m) \leq Y_i \leq Y_1 +i - 1 - \frac{\bar{t}_0}{m} g(i/m)
\end{align}
\end{proof}

\paragraph{Quantification of approximation error.} The sequence $Y_i$ only approximates well $\left\lceil\frac{B^{(i)}_{t_i}}{m-i}\right\rceil$ as long as it stays positive.

\begin{lemma}\label{lem:error_b_i_over_m_m_i_vs_y_i}
    $\left\lceil\frac{B^{(i)}_{t_i}}{m-i}\right\rceil - Y_i \leq \cO(i)$
\end{lemma}
\begin{proof}
    By application of \cref{Z_dynamic} on $B^{(i)}_t$
    \begin{align}
        B^{(i)}_{t_i}
        & = B^{(i)}_{t_{i-1}} - \left\lceil\frac{B^{(i)}_{t_{i-1}}}{m-i+1}\right\rceil + (m-i)\left(1+\left\lfloor\frac{t_i}{m}\right\rfloor-\left\lceil\frac{t_{i-1}}{m}\right\rceil\right) - (t_i - t_{i-1})
    \end{align}
    Thus, using the definition of $Y_i$ in \cref{eq:def_Y_i_adv},
    \begin{align}
        \left\lceil\frac{B^{(i)}_{t_i}}{m-i}\right\rceil - Y_i
        & = \frac{B^{(i)}_{t_i}}{m-i} - Y_i + \cO(1)\\
        & = \frac{1}{m-i}\left(B^{(i)}_{t_{i-1}} - \left\lceil\frac{B^{(i)}_{t_{i-1}}}{m-i+1}\right\rceil\right) + \left(1+\left\lfloor\frac{t_i}{m}\right\rfloor-\left\lceil\frac{t_{i-1}}{m}\right\rceil\right) - \frac{t_i - t_{i-1}}{m-i} - Y_i \nonumber\\ 
        &+\cO(1)\\
        & = (t_i - t_{i-1}) \left(\frac{1}{m} - \frac{1}{m-i}\right) - \left(Y_{i-1} - (\tilde{t}_i - \tilde{t}_{i-1}) \left(\frac{1}{m} - \frac{1}{m-i}\right) + 1\right) \nonumber\\
        &+ \frac{B^{(i)}_{t_{i-1}}}{m-i+1} + 1 +\cO(1)\\
        & = \left\lceil\frac{B^{(i)}_{t_{i-1}}}{m-i+1}\right\rceil -Y_{i-1}  + (t_i - t_{i-1} -\tilde{t}_i + \tilde{t}_{i-1}) \left(\frac{1}{m} - \frac{1}{m-i}\right) + \cO(1)
    \end{align}
    By induction
    \begin{align}
        \left\lceil\frac{B^{(i)}_{t_i}}{m-i}\right\rceil - Y_i
        & = \sum_{j=2}^{i} (t_j - t_{j-1} -\tilde{t}_j + \tilde{t}_{j-1}) \left(\frac{1}{m} - \frac{1}{m-j}\right) + \cO(i)\\
        & = \sum_{j=2}^{i}\frac{(t_j-t_{j-1})-(\tilde{t}_{j}-\tilde{t}_{j+1})}{m}- \sum_{j=2}^{i}\frac{t_j-t_{j-1}-\tilde{t}_{j}+\tilde{t}_{j-1}}{m-j}+\cO(i)\\
        &= \frac{(t_i-t_{1})-(\tilde{t}_{i}-\tilde{t}_{1})}{m}- \sum_{j=2}^{i}\frac{\left((t_{j-1} - \tilde{t}_{j-1}) \left(1 + \frac{1}{m-1}\right) + 1\right)+\tilde{t}_{j-1}-t_{j-1}}{m-j}\nonumber\\
        &+\cO(i) ~~~~~~~~~~~~~~~~~~~~~~~~~~~~~~~~~~~~~~~~~~~~~~~~~~~~~~~\text{by \cref{eq:error_t_tilde_i_intermediary}}\\
        &= \frac{t_i - t_1 - \tilde{t}_i + \tilde{t}_1}{m} - \sum_{j=2}^{i} \frac{t_{j-1} - \tilde{t}_{j-1}+1}{(m-j)(m-1)} + \cO(i) & \\
        & \leq \cO(i) ~~~~~~~~~~~~~~~~~~~~~~~~~~~~~~~~~~~~~~~~~~~~~~~~~~~~~~~ \text{by \cref{lem:error_t_tilde_t}}
    \end{align}

\end{proof}
\begin{lemma}
    $i^{*}=\lfloor \alpha^* m\rfloor$
\end{lemma}
\begin{proof}
The objective is to find $i^*$ such that 
    \begin{align}\label{eq:def_of_i_star}
        Y_1 + i^*-1 - \frac{\bar{t}_0}{m}g((i^*+1)/m) < 0 \leq Y_1 + i^*-1 - \frac{\bar{t}_0}{m}g(i^*/m)
    \end{align}
Let define $\alpha^* \in (1/m,1)$ the solution  of 
\begin{align}
        \frac{\bar{t}_0}{m} \int_{\frac{1}{m}}^\alpha \frac{z}{1-z}e^z {\rm d}z - m\alpha^* - Y_1 = 0\,,
\end{align}
then $i^* = \lfloor \alpha^* m\rfloor$ satisfies \cref{eq:def_of_i_star}.
\end{proof}
\begin{lemma}\label{lemma:sum_B_i_over_m}
    \begin{align}
        \sum_{i=1}^{m-1} \left\lceil\frac{B^{(i)}_i}{m-i}\right\rceil = \lfloor \alpha^* m \rfloor \left\lceil\frac{B^{(1)}_{t_1}}{m}\right\rceil - \bar{t}_0 \int_{\frac{1}{m}}^{\alpha^*} g(x){\rm d}x+ \frac{g(\alpha^*) - g(1/m)}{m} + \cO(m^2)
    \end{align}
    where $\alpha^* \in (1/m,1)$ the solution  of $$\frac{\bar{t}_0}{m} \int_{\frac{1}{m}}^\alpha \frac{z}{1-z}e^z {\rm d}z - m\alpha^* - Y_1 = 0$$  
    and $$g(x) = \frac{x(\alpha^*-x)}{1-x}\left(1+\frac{1}{m-1}\right)^{mx}$$
\end{lemma}

\begin{proof}
    Let define $i^* = \lfloor \alpha^* m\rfloor$. Then,
    \begin{align}
        \sum_{i=1}^{m-1} \left\lceil\frac{B^{(i)}_i}{m-i}\right\rceil 
        & = \sum_{i=1}^{i^*} Y_i + \sum_{i=1}^{i^*} \left\lceil\frac{B^{(i)}_i}{m-i}\right\rceil - Y_i + \sum_{i=i^*+1}^{m-1} \left\lceil\frac{B^{(i)}_i}{m-i}\right\rceil \\
        \sum_{i=1}^{m-1} \left\lceil\frac{B^{(i)}_i}{m-i}\right\rceil &= \sum_{i=1}^{i^*} Y_i+ \sum_{i=1}^{i^{*}}\left(\frac{t_i - t_1 - \tilde{t}_i + \tilde{t}_1}{m} - \sum_{j=2}^{i} \frac{t_{j-1} - \tilde{t}_{j-1}}{(m-j)(m-1)} + \cO(i)\right)\nonumber\\
        &+ (m-1)(m-1-i^{*})-\sum_{i=i^{*}+1}^{m-1} (t_i ~{\rm mod}~m)\\
        \sum_{i=1}^{m-1} \left\lceil\frac{B^{(i)}_i}{m-i}\right\rceil &= \sum_{i=1}^{i^*} Y_i+ i^{*}\frac{\tilde{t}_1-t_1}{m}+ \sum_{i=1}^{i^{*}}\frac{t_i - \tilde{t}_i}{m}- \sum_{i=1}^{i^{*}}\sum_{j=2}^{i}\frac{t_{j-1} - \tilde{t}_{j-1}}{(m-j)(m-1)}+ \cO((i^{*})^{2})\nonumber\\
        &+ (m-1)(m-1-i^{*})- \sum_{i=i^{*}+1}^{m-1} (t_i ~{\rm mod}~m)\\
        \sum_{i=1}^{m-1} \left\lceil\frac{B^{(i)}_i}{m-i}\right\rceil &= \sum_{i=1}^{i^*} Y_i+i^{*}\frac{\tilde{t}_1-t_1}{m}+\sum_{i=1}^{i^{*}}\frac{t_i - \tilde{t}_i}{m}- \sum_{i=1}^{i^*}\sum_{j=2}^{i}\frac{1}{(m-j)} \left(\left(1+\frac{1}{m-1}\right)^{j-1} -1\right)
         \\
        &+(m-1)(m-1-i^{*})+\cO(m^2) ~~~~~~~~~~~~~~~~~~~~~\text{using}~i^{*}=\lfloor\alpha^{*} ~m\rfloor~\text{and~\cref{eq:error_t_tilde_i_intermediary_bis} }\nonumber\\
       \end{align}
       \begin{align}
        \sum_{i=1}^{m-1} \left\lceil\frac{B^{(i)}_i}{m-i}\right\rceil &= \sum_{i=1}^{i^*} Y_i+i^{*}\frac{\tilde{t}_1-t_1}{m}+\sum_{i=1}^{i^{*}}\frac{t_i - \tilde{t}_i}{m}- \frac{1}{(m-1)}\sum_{i=1}^{i^*}\sum_{j=2}^{i} \left(\left(1+\frac{1}{m-1}\right)^{j-1} -1\right)\nonumber\\
        &+(m-1)(m-1-i^{*})+\cO(m^2)\\
        \sum_{i=1}^{m-1} \left\lceil\frac{B^{(i)}_i}{m-i}\right\rceil &= \sum_{i=1}^{i^*} Y_i+ i^{*}\frac{\tilde{t}_1-t_1}{m}+\frac{i^*(i^*+1)}{2(m-1)}-\frac{1}{(m-1)}\sum_{i=1}^{i^*}\sum_{j=2}^{i} \left(\left(1+\frac{1}{m-1}\right)^{j-1} \right)\nonumber\\
        &-\frac{i^*}{m-1}+(m-1)(m-1-i^{*})+\sum_{i=1}^{i^{*}}\frac{t_i - \tilde{t}_i}{m}+\cO(m^2)\\
        \sum_{i=1}^{m-1} \left\lceil\frac{B^{(i)}_i}{m-i}\right\rceil &= \sum_{i=1}^{i^*} Y_i+ i^{*}\frac{\tilde{t}_1-t_1}{m}+\sum_{i=1}^{i^{*}}\frac{t_i - \tilde{t}_i}{m}+\frac{i^*(i^*+1)}{2(m-i^{*})}+ (m-1)(m-1-i^{*})+\cO(m^2)\\
       \sum_{i=1}^{m-1} \left\lceil\frac{B^{(i)}_i}{m-i}\right\rceil  &= \sum_{i=1}^{i^*} Y_i+ \cO(m^2)~~~~~~~~~~~~~~~~~\text{by~\cref{lem:error_t_tilde_t}~and~$i^{*}=\lfloor\alpha^{*}m\rfloor$}
       \end{align}

\begin{align}
 \sum_{i=1}^{m-1} \left\lceil\frac{B^{(i)}_i}{m-i}\right\rceil&= \sum_{i=1}^{i^*}\left(Y_1 + i-1 - \frac{\bar{t}_0}{m^2}\sum_{k=2}^i \frac{k}{m-k} \left(1 + \frac{1}{m-1}\right)^k\right) + \cO(m^2)~~~~~~~\text{by~\cref{eq:y_i_exact}}\\
 \sum_{i=1}^{m-1} \left\lceil\frac{B^{(i)}_i}{m-i}\right\rceil&= i^{*}Y_{1}+\frac{i^{*}(i^*+1)}{2}-i^{*}- \frac{\bar{t}_0}{m^2}\sum_{i=1}^{i^{*}}\sum_{k=2}^i \frac{k}{m-k} \left(1 + \frac{1}{m-1}\right)^k + \cO(m^2)\\
 \sum_{i=1}^{m-1} \left\lceil\frac{B^{(i)}_i}{m-i}\right\rceil&= i^{*}Y_1- \frac{\bar{t}_0}{m^2}\sum_{i=1}^{i^{*}}\sum_{k=2}^i \frac{k}{m-k} \left(1 + \frac{1}{m-1}\right)^k + \cO(m^2)~~~~~~~~~~~~\text{by~$i^{*}=\lfloor\alpha^{*}m\rfloor$}\\
 \sum_{i=1}^{m-1} \left\lceil\frac{B^{(i)}_i}{m-i}\right\rceil& = i^*Y_1 - \frac{\bar{t}_0}{m^2} \sum_{k=2}^{i^*} \frac{(i^* - k) k}{m-k} \left(1 + \frac{1}{m-1}\right)^k + \cO(m^2)\\
        \sum_{i=1}^{m-1} \left\lceil\frac{B^{(i)}_i}{m-i}\right\rceil & = i^*Y_1 - \frac{\bar{t}_0}{m} \sum_{k=2}^{i^*} \frac{(\frac{i^*}{m} - \frac{k}{m}) \frac{k}{m}}{1-\frac{k}{m}} \left(\left(1 + \frac{1}{m-1}\right)^m\right)^\frac{k}{m} + \cO(m^2)\\
        \sum_{i=1}^{m-1} \left\lceil\frac{B^{(i)}_i}{m-i}\right\rceil &=i^{*}Y_1- \frac{\bar{t}_0}{m} \sum_{k=2}^{i^*} g\left(\frac{k}{m}\right)+\cO(m^2) ~~~~~~~~~~~\text{with $g(x)=\frac{x(\alpha^{*}-x)}{1-x}\left(1+\frac{1}{m-1}\right)^{mx}$}\\
        \sum_{i=1}^{m-1} \left\lceil\frac{B^{(i)}_i}{m-i}\right\rceil &=i^{*}Y_1- \frac{\bar{t}_0}{m} \sum_{k=2}^{i^*} g\left(\frac{k}{m}\right)+\cO(m^2)\label{eq:135}\\
        \sum_{i=1}^{m-1} \left\lceil\frac{B^{(i)}_i}{m-i}\right\rceil &= i^{*}\left\lceil\frac{B^{(1)}_{t_1}}{m-1}\right\rceil - \bar{t}_0 \int_{\frac{1}{m}}^{\alpha^{*}} g\left(x\right) \mathrm{d}x+ \frac{g(\alpha^{*})-g(1/m)}{m}+\cO(m^2)\label{eq:136}
\end{align}
  Moving from \cref{eq:135} to \cref{eq:136}  arises from approximating a Riemann sum by an integral.
\end{proof}
\subsection{Proof of \cref{theo:balance_is_best}}\label{app:proof_balance_is_best}
\uballadv*
 \begin{proof}
    The proof is based on the adversarial graph design defined in \cref{app:proof_ub_balance} and organized in two steps: 
    \begin{enumerate}
        \item  Showing that the sequence of total budget decreases at least as fast as the one of $\balance$.
        \item  Using this in \cref{eq:18} to show that $\alg(\GT) \leq \balance(\GT) + o(T)$.
    \end{enumerate}

    $\alg$ is assumed to be any matching algorithm, potentially randomized. The matching built by $\alg$ is denoted $\bf x$, the graph $\GT$ is adversarially defined based on $\bx$ as in \cref{app:proof_ub_balance} and we use the rest of the notation defined there. Note that only the choice of nodes in $U_1,\dots, U_{m-1}$ differs from the graph adversarial to $\alg$, not the other quantities such as $t_i$.
    
    For any $i\leq m-1$ and $t\in\iset{T}$, the total budget of $\alg$ is denoted $B^{(i)}_t = \lE\left[\sum_{u\in U_i}x_{u,t}\right]$ where the expectation is taken over the randomness of $\alg$. We denote $B^{(i),{\tt bal}}_{t}$ for the sequence generated by $\balance$ in \cref{app:proof_ub_balance} for comparison.

    \paragraph{Dynamic of $B^{(i)}_{t_i}$.} First, at time $t_0$, by application of \cref{Z_dynamic}, $B_{t_0}^{(0)} = B^{(0),{\tt bal}}_{t_0}$. The key element of the proof is to use that
    \begin{align}
        \forall y\in \lR_+^n, \forall k\leq n,  \sum_{j=1}^k y_{(n-j)} \leq \frac{k}{n}\|y\|_1 ~~~\text{where}~y_{(j)} \geq 0~\text{is the $i^{th}$ largest coordinate of $y$.}
\label{eq:smurf}
    \end{align}
    Thus, by applying \cref{eq:smurf} on $b_{\cdot,t_0}$, after the adversary removes $n-m+1$ nodes to build $U_1$, we have $B_{t_0}^{(1)} \leq B^{(1),{\tt bal}}_{t_0} + m-1$. The term $r_0 = m-1$ comes from the fact $\balance$ does not exactly equalize the budgets (see \cref{lemma:balance_keeps_equal_budgets}) and a randomized balance algorithm could do it more accurately in expectation. By induction, using at each step \cref{Z_dynamic} and \cref{eq:smurf}, we obtain that $\forall i\leq m-1, ~B^{(i)}_{t_i} \leq B^{(i),{\tt bal}}_{t_i} + i (m-i)$

    \paragraph{Showing that $\alg(\GT) \leq \balance(\GT) + o(T)$.} Denoting $i_t$ the phase of the graph to which time step $t$ belongs, it is possible to show that,
    $$\alg(\GT) = t^* + \sum_{i=i_{t^*}+1}^{m-1} (m-i) \left(\left\lfloor\frac{t_i}{m}\right\rfloor - \left\lfloor\frac{t_{i-1}}{m}\right\rfloor\right) + \cO(m) ~~\text{where}~t^* = \max\{t\in\iset{T}: B^{(i_t)}_{t} \geq 1\}$$
    Noting that, $B^{(i_t),{\tt bal}}_{t} \leq B^{(i_t)}_{t} + i_t(m-i_t)$,  leads to $\alg(G^{\rm Th.2}) \leq \balance(G^{\rm Th.2}) + m^2$, which gives, 
    $$\CR^{\rm adv}(\alg, \cG_{T,m}) \leq \CR(\alg, G^{\rm Th.2})  \leq \CR(\balance, G^{\rm Th.2}) + o_{T}(1)$$
where $\CR(\alg, \GT)=\frac{\alg(\GT)}{\opt(\GT)}$.

\end{proof}

\section{Stochastic Case}
\label{appendix:sto_case}

\subsection{Proof of \cref{theo:lb_stochastic_general}}\label{app:proof_lb_stochastic_general}

\lbstochasticgeneral*

  The proof of  \cref{theo:lb_stochastic_general} is based on Wormald's theorem introduced in \cite{Wormald95,Wormald1995DifferentialEF} and is organized as follows:
  \begin{enumerate}
      \item Definition of the evolution of $(Y_k(t))_{k\geq 0}$ and $\greedy(G,t)$.
      \item Proving that $(Y_k(t))_{k\geq 0}$ and $\greedy(G,t)$ satisfy the hypotheses of the Wormald theorem.
      \item Application of Wormald theorem on $(Y_k(t))_{k\geq 0}$ and $\greedy(G,t)$.
  \end{enumerate}

Recall that for all $u\in U$, $b_{u,t}\in \lN$ is given by, 
\begin{equation}
b_{u,t}=\min(K,b_{u,t-1}-x_{u,t}+ \eta_t) ~~~~~\text{with}~b_{u,0}=b_0\geq 1 ~~~\text{and}~K\in \lN.
    \label{budget_dynamic_sto}
\end{equation}
Here $\eta_t$ is a realization of  a Bernoulli random variable  with parameter $\frac{\beta}{n}$ denoted $\cB\left(\frac{\beta}{n}\right)$.  

First, let's introduce some notations, for $k\in\mathbb{N}$ and $t\in [T]$,
\begin{itemize}
    \item $U_k(t)= \{u\in U: b_{u,t} = k\}$ is the set of nodes with budget $k$.
    \item $Y_{k}(t)=|U_k(t)|$ is the number of nodes with budget equals to $k$. 
    \item $\greedy(G,T)= \sum_{u\in U}\sum_{t=1}^T x_{u,t}$ is  the matching size. 
    \item $C(t)=\sum_{k\geq 1}Y_{k}(t)=n-Y_{0}(t)$  is the total number of nodes with budget at least equals to 1. 
\end{itemize}

In order to apply Wormald's theorem, it is necessary to track the evolution of $\greedy(G,T)$. To achieve this, we must precisely quantify the one-step change in $\greedy(G,t)$ for all $t\in \iset{T}$. This crucial step is addressed in the forthcoming lemma,
\begin{lemma}
\label{lemma:Greedy_one_step_change}
For $t\in \iset{T}$, the expectation of the one-step change of $\greedy(G,t)$ is given by, 
$$\mathbb{E}\left[\greedy(G,t+1)-\greedy(G,t)|\greedy(G,t)\right]= 1-\left(1-\frac{a}{n}\right)^{\sum_{k\geq 1}Y_{k}(t)}=1-\left(1-\frac{a}{n}\right)^{n-Y_{0}(t)} $$
\end{lemma}   
\begin{proof}
For $t\in[T]$, the matching size at time $t+1$ is defined as follows, 
\begin{align*}
\greedy(G,t+1)=\greedy(G,t)+\indicator{x_{u,t+1}=1,~ u\in U_k(t+1)}
\end{align*}

Moving to conditional expectation gives,

\begin{align*}
    \mathbb{E}\left[\greedy(G,t+1)-\greedy(G,t)\middle|\greedy(G,t)\right]&= \mathbb{P}\left(x_{u,t+1}=1,~ u\in U_k\middle|\greedy(G,t)\right)\\
         &= 1-\left(1-\frac{a}{n}\right)^{C(t)}\\
         &= 1-\left(1-\frac{a}{n}\right)^{n-Y_{0}(t)}
\end{align*}
\end{proof} 

  Since the evolution of $\greedy(G,t)$ depends on $Y_0$, we need to quantify  the evolution of $Y_k(t), \forall k\in \lN, t\in \iset{T} $. This is done in the subsequent lemma,
\begin{lemma}
\label{lemma:Y_one_step_change}
    For $t\in \iset{T}$, denoting $\Sigma(t)=\frac{1}{pC(t)} (1-(1-p)^{C(t)})$, the expectation of the one-step change of   $Y_{k}$, when matching is built using $\greedy$ algorithm is given by, 

    \begin{align}
    \label{evolution_Y_k_in_mean}
     \begin{cases}
            \mathbb{E}\left[\Delta_0(t)|\mathbf{Y}(t)\right]= -Y_{0}(t)\left[\frac{\beta}{n} (1-p\,\Sigma(t))\right]+Y_{1}(t)(1-\frac{\beta}{n})p\,\Sigma(t)\\
            \mathbb{E}\left[\Delta_1(t)|\mathbf{Y}(t)\right]= -Y_{1}(t)\left[\frac{\beta}{n}(1-p\,\Sigma(t))+ (1-\frac{\beta}{n})p\,\Sigma(t)\right]+Y_{0}(t)\frac{\beta}{n}+Y_{2}(t)(1-\frac{\beta}{n})p\,\Sigma(t)\\
            \mathbb{E}\left[\Delta_k(t)|\mathbf{Y}(t)\right]=\frac{\beta}{n} (1-p\,\Sigma(t))\left[Y_{k-1}(t)-Y_{k}(t)\right]
+\left[Y_{k+1}(t)-Y_{k}(t)\right](1-\frac{\beta}{n})p\,\Sigma(t)~~~~\forall k>1
        \end{cases}
    \end{align}
where $ \forall k\geq 0, ~\Delta_k(t)=Y_{k}(t+1)-Y_{k}(t)$
\end{lemma}
\begin{proof} 
For $t\in \iset{T}$, the evolution of the number of  nodes with budget $k\in \lN$ can be formulated as, 
\begin{align}
\label{Y_k_evol}
Y_{k}(t+1) &=Y_{k}(t)-\sum_{u\in U_k(t)} \indicator{\{\eta_t=1\}\cap \{ x_{u,t}=0\}\cup\left\{\{x_{u,t}=1\} \cap \{\eta_t=0\}\right\}} + \sum_{u\in U_{k-1}(t)} \indicator{\{\eta_t=1\}\cap \{ x_{u,t}=0\}}+\nonumber\\
    &\sum_{u\in U_{k+1}(t)} \indicator{\{x_{u,t}=1\} \cap \{\eta_t=0\}}
\end{align}

We are interested in the conditional expectation of \cref{Y_k_evol}  denoted by $E(t)=\mathbb{E}\left[Y_{k}(t+1)-Y_{k}(t)|\mathbf{Y}(t) \right]$ where $\mathbf{Y}(t)=\left(Y_k(t)\right)_{k\geq 0}$,
\begin{align}
&E(t)\\
&= -\sum_{u\in U_k(t)} \mathbb{P}\left(\left\{\{\eta_t=1\}\cap \{ x_{u,t}=0\}\right\}\cup\left\{ \{x_{u,t}=1\} \cap \{\eta_t=0\}\right\}\middle|\mathbf{Y}(t)\right)\nonumber\\
&+\sum_{u\in U_{k-1}(t)} \mathbb{P}\left( \left\{\{\eta_t=1\}\cap \{ x_{u,t}=0\}\right\}\Big|\mathbf{Y}(t)\right) +\sum_{u\in U_{k+1}(t)} \mathbb{P}\left( \left\{\{x_{u,t}=1\} \cap \{\eta_t=0\}\right\}\Big|\mathbf{Y}(t)\right) \label{eq:139}\\
&= -\sum_{u\in U_k(t)} \mathbb{P}\left(\left\{\{\eta_t=1\}\cap \{ x_{u,t}=0\}\right\}\middle | \mathbf{Y}(t)\right)- \sum_{u\in U_k(t)} \mathbb{P}\left(
\left\{\{\eta_t=0\}\cap \{ x_{u,t}=1\}\right\}\middle | \mathbf{Y}(t)\right) \nonumber\\
&+ \sum_{u\in U_{k-1}(t)} \mathbb{P}\left( \left\{\{\eta_t=1\}\cap \{ x_{u,t}=0\}\right\}\Big|\mathbf{Y}(t)\right) +\sum_{u\in U_{k+1}(t)} \mathbb{P}\left( \left\{\{x_{u,t}=1\} \cap \{\eta_t=0\}\right\}\Big|\mathbf{Y}(t)\right) \label{eq:140}\\
&= -\sum_{u\in U_k(t)}\frac{\beta}{n}\mathbb{P}\left( \{ x_{u,t}=0\}\middle | \mathbf{Y}(t)\right)- \sum_{u\in U_k(t)} (1-\frac{\beta}{n})\mathbb{P}\left(\{ x_{u,t}=1\}
\middle | \mathbf{Y}(t)\right) \nonumber\\
&+\sum_{u\in U_{k-1}(t)} \frac{\beta}{n}\mathbb{P}\left(  \{ x_{u,t}=0\}\Big|\mathbf{Y}(t)\right)+\sum_{u\in U_{k+1}(t)} (1-\frac{\beta}{n})\mathbb{P}\left( \{x_{u,t}=1\} \middle|\mathbf{Y}(t)\right) \label{eq:141}
\end{align}

Moving from \cref{eq:139}  to \cref{eq:140} and then from \cref{eq:140} to  \cref{eq:141} is done using independence. 

To get the final expression of $E(t)$, we need to compute $\mathbb{P}\left[ \{ x_{u,t}=1\} \Big|\mathbf{Y}(t)\right]$. By using   Bayes formula we can see that, 
\begin{align*}
    \mathbb{P}\left[ \{ x_{u,t}=1\} \Big|\mathbf{Y}(t)\right]&=  \mathbb{P}\left[ \{ (u,t)\in G\} \Big|\mathbf{Y}(t)\right]\mathbb{P}\left[ \{ x_{u,t}=1\} \Big|\mathbf{Y}(t), (u,t)\in G \right]\\
   &=p\,\mathbb{P}\left[ \{ x_{u,t}=1\} \Big|\mathbf{Y}(t), (u,t)\in G \right]
\end{align*}
Now, let's compute $\mathbb{P}\left[ \{ x_{u,t}=1\} \Big|\mathbf{Y}(t), (u,t)\in G \right]$, 
\begin{align*}
    \mathbb{P}\left[ \{ x_{u,t}=1\} \Big|\mathbf{Y}(t), (u,t)\in G \right]&= \sum_{i=1 }^{C(t)}\mathbb{P}\left[ \{ x_{c,t}=1, c\in[i], B_c(t)\geq1 \} \Big|\mathbf{Y}(t), (u,t)\in G \right]\\
    &= \sum_{i=1 }^{C(t)}\mathbb{P}\left[ \{ x_{c,t}=1 \} \Big| c\in[i], B_c(t)\geq1 ,\mathbf{Y}(t), (u,t)\in G \right]  \\
    &\mathbb{P}\left[ c\in[i], B_c(t)\geq1 \Big| \mathbf{Y}(t), (u,t)\in G \right] \\
    &=\sum_{i=1 }^{C(t)} \frac{1}{i} \mathbb{P}\left[ c\in[i], B_c(t)\geq1 \Big| \mathbf{Y}(t), (u,t)\in G \right]\\
    &= \sum_{i=1}^{C(t)} \frac{1}{i}\, {C(t)-1\choose i-1} p^{i-1}\,(1-p)^{C(t)-i}\\
    &= \underbrace{\frac{1}{pC(t)} (1-(1-p)^{C(t)})}_{\Sigma(t)}
\end{align*}

Thus, we get 
$$ \mathbb{P}\left[ \{ x_{u,t}=1\} \Big|\mathbf{Y}(t)\right]=\frac{1}{C(t)} (1-(1-p)^{C(t)})$$

  Due to $\greedy$ algorithm,  here the choice of $u$ inside the probabilities doesn't depend on $U_k$, so putting everything together in $E(t)$, and distinguishing cases where $k=0, k=1$ and $k\geq 2$, we get, 
\begin{align*}
    \begin{cases}
        \mathbb{E}\left[\Delta_0(t)|\mathbf{Y}(t)\right]= -Y_{0}(t)\left[\frac{\beta}{n} (1-p\,\Sigma(t))\right]+Y_{1}(t)(1-\frac{\beta}{n})p\,\Sigma(t)\\
            \mathbb{E}\left[\Delta_1(t)|\mathbf{Y}(t)\right]= -Y_{1}(t)\left[\frac{\beta}{n} (1-p\,\Sigma(t))+ (1-\frac{\beta}{n})p\,\Sigma(t)\right]+Y_{0}(t)\frac{\beta}{n}+Y_{2}(t)(1-\frac{\beta}{n})p\,\Sigma(t)\\
            \mathbb{E}\left[\Delta_k(t)|\mathbf{Y}(t)\right]= \frac{\beta}{n} (1-p\,\Sigma(t))\left[Y_{k-1}(t)-Y_{k}(t)\right]
+\left[Y_{k+1}(t)-Y_{k}(t)\right](1-\frac{\beta}{n})p\,\Sigma(t)~~~~\forall k>1
        \end{cases}
    \end{align*}
  where $ \forall k\geq 0, ~\Delta_k(t)=Y_{k}(t+1)-Y_{k}(t)$.  
  
\end{proof}

  Before establishing the hypotheses of Wormald's theorem, we introduce a technical lemma,
\begin{lemma}
\label{lemma:technical_lemma}
    For $n>0$, $a\leq n/2$ and $0\leq w\leq 1$, 
    $$0\leq e^{-aw}-\left(1-\frac{a}{n}\right)^{nw}\leq \frac{a}{ne}$$
\end{lemma}
\begin{proof}
    Using the following inequalities: $1-x \geq e^{-x-x^2}$ for $x\leq \frac{1}{2}$ and $1-x\leq e^{-x}$ for $x\geq 0$, we obtain $e^{-aw}\left(1-\frac{a^2 w}{n}\right)\leq \left(1-\frac{a}{n}\right)^{nw}\leq e^{-aw} $.  The result follows by rearranging terms and using that $aw e^{-aw}\leq 1/e$. 
    
\end{proof}

  To apply Wormald's theorem \cite{Wormald1995DifferentialEF} in our model, three key hypotheses need to be met: the boundedness hypothesis, the Lipschitz hypothesis, and the trend hypothesis. These hypotheses will be established in the following lemmas for both $\greedy(G,t)$ and $Y_k(t)$. 

\begin{lemma}
\label{lemma:lipschitz_hypothesis}
$\forall k\geq 0$, let $-\epsilon<\tau< \frac{T}{n}+\epsilon$ and  $-\epsilon< z_k<1+\epsilon$  where $\epsilon>0$. The  functions $f_k(\tau)$ and $j_0(\tau)$ defined as follows,  
\begin{align*}
    f_{k}(\tau) &=\begin{cases}
-z_{0}(\tau)\beta + \frac{z_{1}(\tau)}{1-z_{0}(\tau)}(1-e^{-a+az_{0}(\tau)}) &~~\text{for}~~k=0\\
(z_{k-1}(\tau)-z_{k}(\tau))\beta +(z_{k+1}(\tau)-z_{k}(\tau))\frac{1-e^{-a+az_{0}(\tau)}}{1-z_{0}(\tau)} & ~~\text{for}~~k\geq0
\end{cases}\\
 j_0(\tau)&=1-e^{-a(1-z_0(\tau))}
\end{align*}

are Lipschitz with a constant $L= (\beta+a)(1+\epsilon)$ and $L'=ae^{a\epsilon}$ respectively.
\end{lemma}
\begin{proof}
The proof is done for $k=0$ and remains the same for $k\geq 1$. Let $-\epsilon< z_0< 1+\epsilon$, $-\epsilon< z_1< 1+\epsilon$, $-\epsilon<\tau< \frac{T}{n}+\epsilon$ and  $-\epsilon<\tau'< \frac{T}{n}+\epsilon$. 
\begin{align*}
|f_0(\tau)-f_0(\tau')|&= \left|-z_{0}(\tau)\beta + \frac{z_{1}(\tau)}{1-z_{0}(\tau)}(1-e^{-a+az_{0}(\tau)})+z_{0}(\tau')\beta - \frac{z_{1}(\tau')}{1-z_{0}(\tau')}(1-e^{-a+az_{0}(\tau')}) \right| \\
&\leq \beta  \left|z_{0}(\tau')-z_{0}(\tau)\right|+ \left |\frac{z_{1}(\tau)}{1-z_{0}(\tau)}(1-e^{-a+az_{0}(\tau)})+\frac{z_{1}(\tau')}{1-z_{0}(\tau')}(1-e^{-a+az_{0}(\tau')}) \right|\\
&\leq \beta  \left|z_{0}(\tau')-z_{0}(\tau)\right|+ a|z_{1}(\tau)+z_{1}(\tau')| ~~~~~\text{using}(~1-e^{-ax}\leq ax)
z\end{align*}
Thus we get, 
 $$|f_0(\tau)-f_0(\tau')|\leq (\beta+a+2)(1+\epsilon) |\tau-\tau'|$$
Therefore, we proved that $f_0$ is $L$-Lipschitz  with $L= (\beta+a+2)(1+\epsilon) $.  Now, let's proceed to prove that $j_0$ is Lipschitz, 
\begin{align*}
    |j_0(\tau)-j_0(\tau')|&= |e^{-a(1-z_0(\tau))}-e^{-a(1-z_0(\tau'))}|\\
&=e^{-a(1-z_0(\tau))} |1-e^{a(z_0(\tau')-z_0(\tau))}|\\
&\leq e^{-a(1-z_0(\tau))} a|z_0(\tau')-z_0(\tau)| ~~~~~~\text{using $1-e^{ax}\leq -ax$}\\
&\leq e^{a\epsilon} a|\tau-\tau'|
\end{align*}
Hence $j_0$ is $L'$-Lipschitz with $L'=e^{a\epsilon} a$.

\end{proof}

  The next lemma proves the trend hypothesis, 
\begin{lemma}
\label{lemma:trend_hypothesis}
    For $t\in \iset{T}$  the functions $f_{k}\left(\frac{t}{n}, \frac{Y_0(t)}{n},\hdots,\frac{Y_K(t)}{n} \right)$ and $j(\frac{t}{n},\frac{Y_0(t)}{n})$ are given by,
    \begin{align*}
f_{k}&=\begin{cases}
-\frac{Y_0(t)\,\beta}{n}(1-\frac{1}{n-Y_0(t)}(1-e^{-a(1-\frac{Y_{0}(t)}{n})}) + \frac{Y_{1}(t)(n-a)}{n}\frac{1}{n-Y_0(t)}(1-e^{-a(1-\frac{Y_{0}(t)}{n})}) &\text{for}~~k=0\\
\frac{-Y_{1}(t)}{n}\left[a(1-\frac{1}{n-Y_0(t)}(1-e^{-a(1-\frac{Y_{0}(t)}{n})})+ \frac{(n-a)}{n-Y_0(t)}(1-e^{-a(1-\frac{Y_{0}(t)}{n})})\right]\\
+\frac{Y_{2}(t)(n-a)}{n(n-Y_0(t))}(1-e^{-a(1-\frac{Y_{0}(t)}{n})})
 +\frac{Y_{0}(t)a}{n}&\text{for}~~k=1\\
 \frac{-Y_{k}(t)}{n}\left[a(1-\frac{1}{n-Y_0(t)}(1-e^{-a(1-\frac{Y_{0}(t)}{n})})+ \frac{(n-a)}{n-Y_0(t)}(1-e^{-a(1-\frac{Y_{0}(t)}{n})})\right]\\
+\frac{Y_{k+1}(t)(n-a)}{n(n-Y_0(t))}(1-e^{-a(1-\frac{Y_{0}(t)}{n})})
 +\frac{Y_{k-1}(t)a}{n}(1-\frac{1}{n-Y_0(t)}(1-e^{-a(1-\frac{Y_{0}(t)}{n})})&\text{for}~k\leq K-1\\
\frac{Y_{k-1}(t)\,\beta}{n}(1-\frac{1}{n-Y_0(t)}(1-e^{-a(1-\frac{Y_{0}(t)}{n})}) - \frac{Y_{k}(t)(n-a)}{n}\frac{1}{n-Y_0(t)}(1-e^{-a(1-\frac{Y_{0}(t)}{n})}) &\text{for}~~k=K
 \end{cases}\\
 j&=1-e^{-a(1-\frac{Y_0(t)}{n})}
    \end{align*}
and we have for all $k\geq 1$,
 \begin{align}
 \left|\mathbb{E}(\greedy(G,t+1)-\greedy(G,t)|\greedy(G,t))-j\left(\frac{t}{n}, \frac{Y_0(t)}{n}\right)\right|&\leq \frac{a}{en}\label{trend_G}\\
     \left|\mathbb{E}(Y_{k}(t+1)-Y_{k}(t)|\mathbf{Y}(t))-f_k\left(\frac{t}{n}, \frac{Y_0(t)}{n},\hdots,\frac{Y_k(t)}{n},\hdots \right)\right|&\leq \frac{a}{en}\label{trend_Y}
 \end{align}
    
\end{lemma}
\begin{proof}
    Let's prove \cref{trend_Y} for $k=0$ ( the proof is the same for $k\geq 1$),
    \begin{align*}
    M_0&=\left|\mathbb{E}(Y_{0}(t+1)-Y_{0}(t)|\mathbf{Y}(t))-f_0\left(\frac{t}{n}, \frac{Y_0(t)}{n},\frac{Y_1(t)}{n} \right)\right|\\
        M_0&\leq\left|\frac{-Y_0(t)\beta}{n(n-Y_0(t))}(1-(1-\frac{a}{n})^{n-Y_{0}(t)}-1+e^{-a(1-\frac{Y_{0}(t)}{n})})\right|\\
        &+\left|\frac{Y_1(t)(n-\beta)}{n(n-Y_0(t))}(1-(1-\frac{a}{n})^{n-Y_{0}(t)}-1+e^{-a(1-\frac{Y_{0}(t)}{n})})\right|\\
        &\leq \frac{a}{ne}~~~~~~(\text{using \cref{lemma:technical_lemma}})
    \end{align*}
Let's now prove \cref{trend_G},
    \begin{align*}
    P&=\left|\mathbb{E}(\greedy(G,t+1)-\greedy(G,t)|\greedy(G,t))-j\left(\frac{t}{n}, \frac{Y_0(t)}{n}\right)\right|\\
        P&=\left|e^{-a(1-\frac{Y_{0}(t)}{n})}-((1-\frac{a}{n})^{n-Y_{0}(t)})\right|\\
        &\leq \frac{a}{ne}~~~~~~(\text{using \cref{lemma:technical_lemma}})
    \end{align*}
\end{proof}

  The following lemma shows the Boundness hypothesis, 
\begin{lemma}
\label{boundnesshypothesis}
For $t\in \iset{T}$, $k\geq 0$, 
\begin{align*}
|\greedy(G,t+1)-\greedy(G,t)|&\leq \beta'\\
    |Y_k(t+1)-Y_k(t)|&\leq \beta
\end{align*}
with $\beta, \beta'>0$. 
\end{lemma}
\begin{proof}
   For $t\in \iset{T}$ and $k\geq 0$, 
   $$|\greedy(G,t+1)-\greedy(G,t)|=\indicator{\{x_{u,t+1}=1, u\in U_k(t+1)\}}\leq 1$$
   Hence we have $\beta'=1$. 

As seen previously, $Y_k(t)$ is the number of nodes with budget equals to $k$ at time $t$. So, by the nature of the matching process, 
$$ |Y_{k}(t+1)-Y_{k}(t)|\leq 1$$
Hence $\beta=1$. 

\end{proof}

  In the following lemma we  approximate with high probability $Y_k(t),\forall k\geq 0, t\in\iset{T}$ by the solution of a system of differential equations, 
\begin{lemma}
\label{lemma:Y_approximation}
With probability $1-\cO\left(n^{1/4} \exp(-a^3 n^{1/4})\right)$, 
$$Y_k(T)=nz_k(T/n)+\cO(n^{3/4})~~~~\text{for}~k\geq 0$$ 
$\forall~\tau\in [\frac{1}{n},\frac{T}{n}]$, $(z_0,\hdots,z_K)$ is the solution of the following system, 
\begin{align}
    \begin{cases}
        \dot{z}_{0}(\tau)=-z_{0}(\tau)\beta + \frac{z_{1}(\tau)}{1-z_{0}(\tau)}(1-e^{-a+az_{0}(\tau)}) \hspace{1cm } & \text{for} ~k=0\\
        \dot{z}_{k}(\tau)= (z_{k-1}(\tau)-z_{k}(\tau))\beta +(z_{k+1}(\tau)-z_{k}(\tau))\frac{1-e^{-a+az_{0}(\tau)}}{1-z_{0}(\tau)} & \text{for} ~ 1\leq k \leq K-1\\
        \dot{z}_{k}(\tau)=\beta\,z_{k-1}(\tau)-z_{k}(\tau)\frac{1-e^{-a(1-z_0(\tau))}}{1-z_0(\tau)} &\text{for}~k= K\\
        \sum_{k=0}^K z_{k}(\tau)= 1
    \end{cases}
    \label{sysdiffequaZk}
\end{align}
\end{lemma}

\begin{proof}
For $\frac{1}{n}\leq\tau\leq \frac{T}{n}$, let's consider the  normalized random variable $Z_{k}(\tau)=\frac{Y_{k}(\tau\,n)}{n}$ and $\mathbf{Z}(\tau)=\left(Z_{k}(\tau)\right)_{k\geq0}$. The conditional expectation of the one-step change of $Z_k(\tau)$ for different values of $k$ is given by, 
\begin{itemize}
    \item For $k=0$,
    \begin{align*}
        \frac{\mathbb{E} \left[Z_{0}(\tau+\frac{1}{n})-Z_{0}(\tau)|\mathbf{Z}(\tau)\right]}{1/n}&=  \frac{\mathbb{E} \left[Y_{0}(\tau\,n+1)/n-Y_{0}(\tau\,n)/n|\mathbf{Y}(\tau\,n)/n  \right]}{1/n}\\
    &= \frac{\mathbb{E}\left[-Z_{0}(\tau)\left[\frac{\beta}{n} (1-\,\frac{1}{n-nZ_{0}(\tau)} (1-(1-p)^{n-nZ_{0}(\tau)}))\right]\right]}{1/n}\\
    & + \frac{\mathbb{E}\left[ Z_{1}(\tau)(1-\frac{\beta}{n})\,\frac{1}{n-nZ_{0}(\tau)} (1-(1-p)^{n-nZ_{0}(\tau)})\right]}{1/n}
    \end{align*}
    when $n\to +\infty$, we get, 
   $$
   \dot{z}_{0}(\tau)=-z_{0}(\tau)\beta + \frac{z_{1}(\tau)}{1-z_{0}(\tau)}(1-e^{-a+az_{0}(\tau)})$$
    \item For  $k=1$, 
    \begin{align*}
        \frac{\mathbb{E} \left[Z_{1}(\tau+\frac{1}{n})-Z_{1}(\tau)\middle|\mathbf{Z}(\tau)\right]}{1/n}&=  \frac{\mathbb{E} \left[Y_{1}(\tau\,n+1)/n-Y_{1}(\tau\,n)/n\middle|\mathbf{Y}(\tau\,n)/n  \hspace{0.1cm}\forall k\geq 1\right]}{1/n}\\
    &= \frac{\mathbb{E}\left[-Z_{1}(\tau)\left[\frac{\beta}{n} (1-\,\frac{1}{n-nZ_{0}(\tau)} (1-(1-p)^{n-nZ_{0}(\tau)}))\right]\right]}{1/n}\\
    & + \frac{\mathbb{E}\left[ -Z_{1}(\tau)(1-\frac{\beta}{n})\frac{1}{n-nZ_{0}(\tau)}(1-(1-p)^{n-nZ_{0}(\tau)}))\right]}{1/n}\\
    &+ \frac{\mathbb{E}\left[Z_{0}(\tau)\frac{\beta}{n} +Z_{2}(\tau)(1-\frac{\beta}{n})\frac{1}{n-nZ_{0}(\tau)} (1-(1-p)^{n-nZ_{0}(\tau)})\right]}{1/n}
    \end{align*}
  when $n\to +\infty$ we get, 
  $$
  \dot{z}_{1}(\tau)=\beta (z_0(\tau)-z_1(\tau))+(z_2(\tau)-z_1(\tau))\frac{1-e^{-a+az_{0}(\tau)}}{1-z_{0}(\tau)}$$
  \item $k\geq 2$, 
  \begin{align*}
        \frac{\mathbb{E} \left[Z_{k}(\tau+\frac{1}{n})-Z_{k}(\tau)|\mathbf{Z}(t)\right]}{1/n}&=  \frac{\mathbb{E} \left[Y_{k}(\tau\,n+1)/n-Y_{k}(\tau\,n)/n|\mathbf{Y}(\tau\,n)/n  \right]}{1/n}\\
    &= \frac{\mathbb{E}\left[-Z_{k}(\tau)\left[\frac{\beta}{n} (1-\,\frac{1}{n-nZ_{0}(\tau)} (1-(1-p)^{n-nZ_{0}(\tau)}))\right]\right]}{1/n}\\
    & + \frac{\mathbb{E}\left[ -Z_{k}(\tau)(1-\frac{\beta}{n})\frac{1}{n-nZ_{0}(\tau)}(1-(1-p)^{n-nZ_{0}(\tau)}))\right]}{1/n}\\
    &+ \frac{\mathbb{E}\left[Z_{k-1}(\tau)\frac{\beta}{n} (1-\,\frac{1}{n-nZ_{0}(\tau)} (1-(1-p)^{n-nZ_{0}(\tau)})) \right]}{1/n}\\
    &+\frac{\mathbb{E}\left[Z_{k+1}(\tau)(1-\frac{\beta}{n})\frac{1}{n-nZ_{0}(\tau)}(1-(1-p)^{n-nZ_{0}(\tau)}))\right]}{1/n}
    \end{align*}
     when $n\to +\infty$ we get, 
     $$\dot{z}_{k}(\tau)= (z_{k-1}(\tau)-z_{k}(\tau))\beta +(z_{k+1}(\tau)-z_{k}(\tau))\frac{1-e^{-a+az_{0}(\tau)}}{1-z_{0}(\tau)}$$
\end{itemize}
Applying the Wormald theorem \citep{Wormald95,Wormald1995DifferentialEF}, with the domain $D$ defined by $-\epsilon<\tau<\frac{T}{n}+\epsilon$, $-\epsilon<z_k< 1+\epsilon$, for $\epsilon>0$. And taking $\beta=1$ for the boundeness hypothesis (see  \cref{boundnesshypothesis}),  $\Lambda_1=a/(e\,n)$ for the trend hypothesis (see \cref{lemma:trend_hypothesis}).  The Lipschitz hypothesis is satisfied with Lipschitz constant $L=(\beta+a)(1+\epsilon)$  (see \cref{lemma:lipschitz_hypothesis}). Setting  $\lambda=a\,n^{-1/4}$, the Wormald theorem  gives with probability  $1-\cO\left(n^{1/4} \exp(-a^3 n^{1/4})\right)$, 
$$Y_k(T)=nz_k(T/n)+\cO(n^{3/4})~~~~\text{for}~k\geq 0$$
with $(z_0,\hdots,z_K)$ the solution of the following system, 
$$
\begin{cases}
        \dot{z}_{0}(\tau)=-z_{0}(\tau)\beta + \frac{z_{1}(\tau)}{1-z_{0}(\tau)}(1-e^{-a+az_{0}(\tau)}) \hspace{1cm } & \text{for} ~k=0\\
        \dot{z}_{k}(\tau)= (z_{k-1}(\tau)-z_{k}(\tau))\beta +(z_{k+1}(\tau)-z_{k}(\tau))\frac{1-e^{-a+az_{0}(\tau)}}{1-z_{0}(\tau)} & \text{for} ~ 1\leq k \leq K-1\\
        \dot{z}_{k}(\tau)=\beta\,z_{k-1}(\tau)-z_{k}(\tau)\frac{1-e^{-a(1-z_0(\tau))}}{1-z_0(\tau)} &\text{for}~k= K\\
        \sum_{k=0}^K z_{k}(\tau)= 1
    \end{cases}$$
\end{proof}

Now we have all the tools to prove \cref{theo:lb_stochastic_general}.
 
\begin{proof}
For $\frac{1}{n}\leq \tau\leq \frac{T}{n}$, let's consider the normalized random variable  $H(\tau)=\frac{\greedy(G,\tau\,n)}{n}$, the conditional expectation of the one-step change of $H(\tau)$ is given by, 
\begin{align*}
    \frac{\mathbb{E}\left[H\left(\tau+\frac{1}{n}\right)-H(\tau)\middle| H(\tau)\right]}{1/n}&= \frac{\mathbb{E}\left[\greedy(G,\tau\,n+1)/n-\greedy(G,\tau\,n)/n \middle| \greedy(G,\tau\,n)/n\right]}{1/n}\\
    &= 1-(1-\frac{a}{n})^{n-n\,Z_{0}(\tau)}
\end{align*}
when $n\to +\infty$ we get, 
$$\dot{h}(\tau)=1-e^{-a\left(1-\,z_{0}(\tau)\right)} $$

  Applying Wormald theorem \citep{Wormald95,Wormald1995DifferentialEF}, we choose the domain $D$ defined by $-\epsilon<\tau<\frac{T}{n}+\epsilon$, $-\epsilon<z_0< 1+\epsilon$ for $\epsilon>0$. We have $\beta'=1$ for the boundedness hypothesis (see \cref{boundnesshypothesis}), $\delta=a/(e\,n)$ for the trend hypothesis (\cref{lemma:trend_hypothesis}). The Lipschitz hypothesis is satisfied with Lipschitz constant $L'=ae^{a\,\epsilon}$. Setting  $\lambda=a\,n^{-1/4}$, the Wormald theorem gives with probability $1-\cO\left(n^{1/4} \exp(-a^3 n^{1/4})\right)$, 
$$\greedy(G,T)=nh(T/n)+\cO(n^{3/4})$$
 where $h(\tau)$ is solution of the following equation,  
$$\dot{h}(\tau)=1-e^{-a(1-z_{0}(\tau))},~~~~~ 1/n\leq \tau\leq T/n$$
and $z_0(\tau)$ as defined in the following system,
\begin{align*}
    \begin{cases}
        \dot{z}_{0}(\tau)=-z_{0}(\tau)\beta + \frac{z_{1}(\tau)}{1-z_{0}(\tau)}(1-e^{-a+az_{0}(\tau)}) \hspace{1cm } & \text{for} ~k=0\\
        \dot{z}_{k}(\tau)= (z_{k-1}(\tau)-z_{k}(\tau))\beta +(z_{k+1}(\tau)-z_{k}(\tau))\frac{1-e^{-a+az_{0}(\tau)}}{1-z_{0}(\tau)} & \text{for} ~ 1\leq k \leq K-1\\
        \dot{z}_{k}(\tau)=\beta\,z_{k-1}(\tau)-z_{k}(\tau)\frac{1-e^{-a(1-z_0(\tau))}}{1-z_0(\tau)} &\text{for}~k= K\\
        \sum_{k=0}^K z_{k}(\tau)= 1
    \end{cases}
\end{align*}

  Since $\greedy(G,T)$ is bounded and thus uniformly integrable, so convergence in probability implies convergence in mean:
$$\underset{n\to +\infty}{\lim}\frac{\lE[\greedy(G,T)]}{n}=h(T/n)$$
\end{proof}

The next theorem  applies an improved version of the Wormald theorem (see \cite{Warnke2019OnWD}, \cite{Enriquez2019DepthFE}) on $\greedy(G,T)$,

\begin{theorem}
\label{theorem:improved wormald}
With probability at least $1-2e^{-a^2 n^{\frac{3}{2}}/8T}$ we have, 
$$\underset{1\leq t\leq T}{\max}~|\greedy(G,t)-nh(t/n)|\leq 3e^{L'T/n}a n^{3/4}$$ 
with $L'=ae^{a\epsilon}$ and $\epsilon>0$, here  $h(\tau)$ is solution of the following equation,  
$$\dot{h}(\tau)=1-e^{-a(1-z_{0}(\tau))}~~~~~~~ 1/n\leq \tau\leq T/n$$
and $z_0(\tau)$ is defined by the following system, 
\begin{align}
    \begin{cases}
        \dot{z}_{0}(\tau)=-z_{0}(\tau)\beta + \frac{z_{1}(\tau)}{1-z_{0}(\tau)}(1-e^{-a+az_{0}(\tau)}) \hspace{1cm } & \text{for} ~k=0\\
        \dot{z}_{k}(\tau)= (z_{k-1}(\tau)-z_{k}(\tau))\beta +(z_{k+1}(\tau)-z_{k}(\tau))\frac{1-e^{-a+az_{0}(\tau)}}{1-z_{0}(\tau)} & \text{for} ~ 1\leq k \leq K-1\\
        \dot{z}_{k}(\tau)=\beta\,z_{k-1}(\tau)-z_{k}(\tau)\frac{1-e^{-a(1-z_0(\tau))}}{1-z_0(\tau)} &\text{for}~k= K\\
        \sum_{k=0}^K z_{k}(\tau)= 1
    \end{cases}
\end{align}
\end{theorem}
\begin{proof}
  Using the normalized random variable  $H(\tau)=\frac{\greedy(G,\tau\,n)}{n}$ with $\frac{1}{n}\leq \tau\leq \frac{T}{n}$, let's compute the conditional expectation of the one-step change of $H(\tau)$, 
\begin{align*}
    \frac{\mathbb{E}\left[H\left(\tau+\frac{1}{n}\right)-H(\tau)\middle| H(\tau)\right]}{1/n}&= \frac{\mathbb{E}\left[\greedy(G,\tau\,n+1)/n-\greedy(G,\tau\,n)/n \middle| \greedy(G,\tau\,n)/n\right]}{1/n}\\
    &= 1-(1-\frac{a}{n})^{n-n\,Z_{0}(\tau)}
\end{align*}
when $n\to +\infty$ we get, 
$$\dot{h}(\tau)=1-e^{-a\left(1-\,z_{0}(\tau)\right)} $$
Applying the non-asymptotic version of the Wormald theorem \citep{Warnke2019OnWD}, we choose the domain $D$ defined by $-\epsilon<\tau<\frac{T}{n}+\epsilon$, $-\epsilon<z_0< 1+\epsilon$ for $\epsilon>0$. We have $\beta'=1$ (\cref{boundnesshypothesis}), $\delta=a/(e\,n)$ for the trend hypothesis(\cref{lemma:trend_hypothesis}). The Lipschitz hypothesis is satisfied with Lipschitz constant $L'=ae^{a\,\epsilon}$(\cref{lemma:lipschitz_hypothesis}). Setting  $\lambda=a\,n^{-1/4}$ we have with probability at least $1-2e^{-a^2 n^{\frac{3}{2}}/8T}$, 
$$\underset{1\leq t\leq T}{\max}~|\greedy(G,t)-nh(t/n)|\leq 3e^{L'T/n}a n^{3/4}$$
\end{proof}

\subsection{Proof of \cref{corr:lb_stochastic_z_star_allk}}\label{app:proof_llb_stochastic_z_star_allk}
\lbstochasticallk*

  The proof is organized as follows:
  \begin{enumerate}
      \item Finding the stationary solution of \cref{sysdiffequaZk}.
      \item Proving that the stationary solution is asymptotically stable.
      \item Proving that $\greedy(G,T)$ converges to a function depending on the stationary solution.
  \end{enumerate}

 The next result gives a general form for the stationary solution of \cref{sysdiffequaZk}, 
\begin{lemma}
\label{reccurence_on_z_k}
    For $\frac{1}{n}\leq \tau \leq \frac{T}{n}$, let $\bar{S}_{z_0^*}=(z_0^*,\hdots,z_k^*,\hdots,z_K^*)$ be the stationary solution of the system, 
    \begin{align}
    \begin{cases}
        \dot{z}_{0}(\tau)=-z_{0}(\tau)\beta + \frac{z_{1}(\tau)}{1-z_{0}(\tau)}(1-e^{-a+az_{0}(\tau)}) \hspace{1cm } & \text{for} ~k=0\\
        \dot{z}_{k}(\tau)= (z_{k-1}(\tau)-z_{k}(\tau))\beta +(z_{k+1}(\tau)-z_{k}(\tau))\frac{1-e^{-a+az_{0}(\tau)}}{1-z_{0}(\tau)} & \text{for} ~ 1\leq k \leq K-1\\
        \dot{z}_{k}(\tau)=\beta\,z_{k-1}(\tau)-z_{k}(\tau)\frac{1-e^{-a(1-z_0(\tau))}}{1-z_0(\tau)} &\text{for}~k= K\\
        \sum_{k=0}^K z_{k}(\tau)= 1
    \end{cases}
    \label{sysdiffequaZk_bis}
\end{align}
$\bar{S}_{z_0^*}$ is unique and satisfies, 
     \begin{align}
         \label{z_k_star}
         z_k^{*}=z_0^*\left(\frac{\beta}{g(z_0^*)}\right)^k~~~\text{for}~0\leq k\leq K
     \end{align}
    where $g(z_0^*)=\frac{1-e^{-a(1-z_0^*)}}{1-z_{0}^*}$. 
\end{lemma}
\begin{proof}
    \cref{z_k_star} is proved by recurrence. 
    For the uniqueness, according to \cref{sysdiffequaZk_bis}, $\bar{S}_{z_0^*}$ satisfies $\sum_{k=1}^{K} z_k^{*}=1$, using \cref{z_k_star} we get that $P(z_0^*)=\sum_{k=1}^{K} z_0^*\left(\frac{\beta}{g(z_0^*)}\right)^k -1$. $P(0)=-1$ and $\underset{z_0\to 1}{\lim}~P(z_0)>0 $. Moreover $P$ is continuous and monotonic, this implies that $P(z_0)=0$ has a unique solution. Thus, $\bar{S}_{z_0^*}$ is unique.

\end{proof}
\begin{remark}
\label{convergence_zk_cond}
  Given that $z_k^{*}$ follows a geometric progression, for convergence, it's essential that $\left|\frac{\beta}{g(z_0^{*})}\right|\leq 1$. Therefore, we will proceed with the remaining proofs under this assumption.
\end{remark}

The following lemma shows that $\bar{S}_{z_0^*}$ is an asymptotically stable stationary solution of \cref{sysdiffequaZk_bis}. 

\begin{theorem}
\label{theo:stationary_solution_stable_all_k}
$\bar{S}_{z_0^*}$ is a an asymptotically stable stationary solution of \cref{sysdiffequaZk_bis}. 
\end{theorem}
\begin{proof}
     Let $Z= \begin{pmatrix}z_0(t)\\ \vdots\\ z_K(t)\end{pmatrix}$, \cref{sysdiffequaZk} can be seen as $\dot{Z}=F(Z)$, where,
\begin{align*}
    F(z_0(t),\hdots, z_K(t))= \left(-\beta\,z_0+ z_1g(z_0),\hdots, \beta\,z_{k-1}(t)-z_{k}(t)\frac{1-e^{-a(1-z_0(t))}}{1-z_0(t)}\right)
\end{align*}
The Jacobian of $F$ at $\Bar{S}_{z_0^*}$ is then given by, 
$$
DF(\Bar{S}_{z_0^*})= \begin{pmatrix}
    -\beta+z_1^*g'(z_0^*) & g(z_0^*) &0 & \ldots& \ldots & 0\\
    \beta+(z_2^*-z_1^*)g'(z_0^*) &-\beta-g(z_0^*)&g(z_0^*)&0&\ldots&0\\
    (z_3^*-z_2^*)g'(z_0^*) &\beta&-\beta-g(z_0^*)&g(z_0^*)&0..&0\\
    \vdots &0&\beta&-\beta-g(z_0^*)&g(z_0^*)&..0
    \\
    \vdots &\vdots&\vdots&\vdots&\vdots&\vdots
    \\
    -z^*_Kg'(z_0^*) &\ldots&\ldots&0&\beta&-g(z_0^*)
\end{pmatrix}
$$
Since proving that $\Bar{S}_{z_0^*}$ is asymptotically stable is equivalent to proving that
the eigenvalues of $DF(\Bar{S}_{z_0^*})$ are non-positives are non-positives \citep{EDO_course_X}. We achieve this using the perturbation method. To do so, we shall write, 
$$DF(\Bar{S}_{z_0^*})= M +uv^\top$$
where $v^\top=(1,0,\ldots,0)$ and,
\begin{align*}
    u^\top&=g'(z_0^*)(z_1^*,z_2^*-z_1^*,\ldots,-z^*_K)\\
&=z_1^*g'(z_0^*)\left(1,\frac{\beta}{g(z_0^*)}-1,(\frac{\beta}{g(z_0^*)}-1)\frac{\beta}{g(z_0^*)},(\frac{\beta}{g(z_0^*)}-1)\Big(\frac{\beta}{g(z_0^*)}\Big)^2,\ldots\right)
\end{align*}
and $M$ is the matrix with $g(z_0^*)$ above the diagonal, $\beta$ below it and its diagonal is $(-\beta,-\beta-g(z_0^*),\ldots,-\beta-g(z_0^*),-g(z_0^*))$.

$$
M= \begin{pmatrix}
    -\beta & g(z_0^*) &0 & \ldots& \ldots & 0\\
    \beta&-\beta-g(z_0^*)&g(z_0^*)&0&\ldots&0\\
   0 &\beta&-\beta-g(z_0^*)&g(z_0^*)&0..&0\\
    \vdots &0&\beta&-\beta-g(z_0^*)&g(z_0^*)&..0
    \\
    \vdots &\vdots&\vdots&\vdots&\vdots&\vdots
    \\
    0 &\ldots&\ldots&0&\beta&-g(z_0^*)
\end{pmatrix}
$$
 Let us denote by $\Pi_M(\lambda)$ the characteristic polynomial of $M$, as a function of $\lambda$, so that
$$
\Pi_{M+uv^\top}(\lambda)=\Pi_M(\lambda)\left(1+v^\top(M-\lambda I)^{-1}u\right)
$$
which implies that eigenvalues of $M+uv^\top$ are either eigenvalues of $M$ or solutions of $$1+v^\top(M-\lambda I)^{-1}u=0.$$
Since 0 is an eigenvalue of $M+uv^\top$, we aim at proving that $1+v^\top(M-\lambda I)^{-1}u=0$ has $K$  non-positives solutions.

 We now claim that the eigenvalues of $M$ are $\mu_j:=-\beta-g(z_0^*)+2\sqrt{\beta g(z_0^*)}\cos(\frac{j\pi}{k+1})$ for $j \in [k]$ and $0$. This is a consequence of standard computations along with Theorem 2.2 of \cite{eig_top}. We also denote by $\omega_j$ the eigenvectors of $M$ associated to $\mu_j$ (and $\omega_0$ to 0) and by $P$ the matrix whose columns are $\omega_0,\omega_1,...$. As a consequence,
\begin{align*}
    q(\lambda)=1+v^\top(M-\lambda I)^{-1}u&= 1+v^\top P\text{diag}\Big(\frac{1}{\mu_j-\lambda}\Big)P^{-1}u\\
    &=1 + \Big(\frac{\omega_{0,1}}{-\lambda},\ldots,\frac{\omega_{j,1}}{\mu_j-\lambda}\Big)P^{-1}u
    \end{align*}
Since we can take any eigenvectors in $P$, we can assume that $\omega_{j,1}\geq 0$ hence it remains to prove that $P^{-1}u$ is a vector with non-negative coordinates. Notice that this vector is the vector of $u$ in the basis formed by the eigenvectors of $M$. 

 The computations of $\omega_j$ are quite standard, and they yield, denoting $\theta=\sqrt{\frac{\beta}{g(z_0^*)}}$, for $1\leq j\leq K$
\begin{align*}
 \omega_j=&\Big(\theta\sin(\frac{j\pi}{K+1}),\hdots, (-\theta)^{t+1}\sin(\frac{(t+1)j\pi}{K+1})+(-\theta)^{t}\sin(\frac{tj\pi}{K+1}),\hdots, (-\theta)^{K+1}\sin(\frac{(K+1)j\pi}{K+1})\\
&+(-\theta)^{K}\sin(\frac{tK\pi}{K+1})\Big)
\end{align*}
As a consequence, $u$ and all the eigenvectors $\omega_j$ are orthogonal to the vectors of ones, which indicates that $u=\sum_j \alpha_j \omega_j$ for some scalar $\alpha_j$. The objective is to prove that they are all positive

 The exact forms of $u$ and $\omega_j$ give, after a few lines of algebra, 
$$\sum_{j=1}^K \alpha_j\sin(m\frac{j\pi}{K+1})=(-\theta)^m, \qquad \forall m \in [K].$$
This system can be rewritten using the Chebyshev polynomials of second kind (denoted by $U_n$) as, 
$$\sum_{j=1}^K \alpha_jU_{m-1}(\cos(\frac{j\pi}{K+1}))\sin(\frac{j\pi}{K+1})=(-\theta)^m, \qquad \forall m \in [K].$$
Hence in a more compact matrix way
it can be seen as,
$$
W\begin{pmatrix}
\alpha_1\sin(\frac{\pi}{K+1}) \\
\vdots\\
\alpha_K\sin(\frac{K\pi}{K+1})
\end{pmatrix} = 
\vec{-\theta}
,
$$
where $\vec{-\theta}=\Big((-\theta )^j\big)_{j 
\in [K]}$ and $W$ is the matrix whose $j$-th column is $$(U_0(\cos(\frac{j\pi}{K+1}),\ldots,U_{K-1}(\frac{j\pi}{K+1}))^\top$$

  We introduce now the following polynomial,  
$$
P_m(X)=\gamma_m\frac{U_k(X)}{X-\cos(\frac{m\pi}{K+1})}=\sum_{j=1}^K\beta_{j,m} U_{j-1}(X)$$
where $$\gamma_m=\frac{1}{\Pi_{j
\neq m} (\cos(\frac{m\pi}{K+1})-\cos(\frac{j\pi}{K+1}))2^{K}}=\frac{1}{2^K}\frac{1}{\Pi_{j\neq m}-2\sin(\frac{m+j}{2}\frac{\pi}{K+1})\sin(\frac{m-j}{2}\frac{\pi}{K+1})}$$
so that the sign of $\gamma_m$ is $(-1)^{m-1}$, $P_m(\cos(\frac{j\pi}{K+1}))=0$ for all $j\neq m$, and $P_m(\cos(\frac{m\pi}{K+1})=1$. We get that  
$$
\alpha_m = \frac{1}{\sin(\frac{m\pi}{K+1})}\sum_{j=1}^K\beta_{j,m}(-\theta)^j.$$
Using the fact that, by definition of $P_m$,
$$U_K(X)=\sum_{j=1}^K\frac{1}{\gamma_m}
\beta_{j,m}U_{j=1}(X)(X-\cos(\frac{m\pi}{K+1}))
$$
and the property of Chebyshev polynomial,
$$
U_k(X)=2XU_{k-2}(X)-U_{k-3},
$$
we can identify the coefficients $
\beta_j$ that satisfy a linear recurrence of order 2 and are defined by
$$\beta_{j}=2\gamma_m\frac{\sin(\frac{(K-j+1)m\pi}{K+1})}{\sin(\frac{m\pi}{K+1})}.$$
It remains to compute $
\alpha_m=\sum \beta_j(-
\theta)^j$, and standard computations yield that 
$$
\alpha_m=-2\gamma_m(-1)^m\frac{(1-\theta^{2})}{1+2\cos(\frac{m\pi}{K+1})\theta+\theta^2}\geq 0$$


Thus, we have proved that $P^{-1}u$ is a vector with non-negative coordinates. Consequently,  $q(\lambda)$ is an increasing function of $\lambda$. As a result,  $M+uv^\top$ has $K$ eigenvalues of negative real part and one eigenvalue equals to zero. From this, we can conclude that $\Bar{S}_{z_0^*}$ is an asymptotically stationary solution of \cref{sysdiffequaZk_bis}. 

\end{proof}

Given the previous results, we can prove \cref{corr:lb_stochastic_z_star_allk}, 

\begin{proof}
Let $h^*(T/n)=\int_{1/n}^{T/n} (1-e^{-a(1-z_0^*)})\mathrm{d}\tau=\frac{(T-1)(1-e^{-a(1-z_0^*)})}{n}$, we have, 
\begin{align*}
   |\greedy(G,T)-nh^*(T/n)|&= |\greedy(G,T)-nh(T/n)+nh(T/n)-nh^*(T/n)|\\
    &\leq \underbrace{\underset{1\leq  t\leq T}{\max}~(|\greedy(G,t)-nh(t/n)|}_{\leq 3e^{L'T/n}a n^{3/4} ~\text{by \cref{theorem:improved wormald}}}+|nh(T/n)-nh^*(T/n)|)
\end{align*}

  Let's focus on $D=|nh(T/n)-nh^*(T/n)|$, for $1\leq T'<T$ and $\delta >0$, 
\begin{align*}
    D&= |nh(T/n)-nh(T'/n+\delta)+nh(T'/n+\delta)-nh^*(T/n)|\\
    &=|nh(T/n)-nh(T'/n+\delta)-nh^*(T/n)+ nh^*(T'/n+\delta)+nh(T'/n+\delta)-nh^*(T'/n+\delta)|\\
    &= \left|n\int_{T'/n+\delta}^{T/n}(e^{-a(1-z_0^{*})}-e^{-a(1-z_0(t))})\mathrm{d}t+ n\int_{1/n}^{T'/n+\delta}(e^{-a(1-z_0^{*})}-e^{-a(1-z_0(t))})\mathrm{d}t\right|\\
    &= \left|ne^{-a(1-z_0^{*})}\int_{T'/n+\delta}^{T/n}(1-e^{a(z_{0}(t)-z_0^{*})})\mathrm{d}t+ ne^{-a(1-z_0^{*})}\int_{1/n}^{T'/n+\delta}(1-e^{a(z_{0}(t)-z_0^{*}})\mathrm{d}t\right|\\
    &\leq ne^{-a(1-z_0^{*})}\left|\int_{T'/n+\delta}^{T/n}(-a(z_{0}(t)-z_0^{*}))\mathrm{d}t+ \int_{1/n}^{T'/n+\delta}(-a(z_{0}(t)-z_0^{*}))\mathrm{d}t\right| ~~~\text{using}~(1-e^{ax}\leq -ax)\\
    &\leq ne^{-a(1-z_0^{*})}a\int_{T'/n+\delta}^{T/n}|z_{0}(t)-z_0^{*}|\mathrm{d}t+ ne^{-a(1-z_0^{*})}a\int_{1/n}^{T'/n+\delta}|z_{0}(t)-z_0^{*}|\mathrm{d}t\\
    &\leq n e^{-a(1-z_0^{*})}a\int_{T'/n+\delta}^{T/n}\epsilon\mathrm{d}t+ ne^{-a(1-z_0^{*})}a\int_{1/n}^{T'/n+\delta}|z_{0}(t)-z_0^{*}|\mathrm{d}t ~~~~~\text{using}~~(\cref{theo:stationary_solution_stable_all_k})\\
    &\leq ne^{-a(1-z_0^{*})}a\epsilon \left(\frac{T-T'}{n}-\delta\right)+ne^{-a(1-z_0^{*})}a\int_{1/n}^{T'/n+\delta}1\mathrm{d}t~~~~~~~\text{using~$0\leq z_0\leq 1$~and~$0\leq z_0^{*}\leq 1$}\\
    &\leq ne^{-a(1-z_0^{*})}a\epsilon \left(\frac{T}{n}\right)+ ne^{-a(1-z_0^{*})}a\left(\frac{T'-1}{n}+\delta\right)(1-\epsilon)\\
    &\leq ne^{-a(1-z_0^{*})}a\epsilon \left(\frac{T}{n}\right)+ ne^{-a(1-z_0^{*})}a\underbrace{\left(\frac{T'-1}{n}+\delta\right)}_{\leq \frac{T}{n^2}}\\
&\leq 2e^{-a(1-z_0^{*})}a \left(\frac{T}{n}\right)~~~~~~\text{choosing}~(\epsilon=\frac{1}{n})
\end{align*}
Thus, 
$$|\greedy(G,T)-nh^*(T/n)|
\leq3e^{L'T/n}a n^{3/4} +2e^{-a(1-z_0^{*})}a \frac{T}{n}$$
  with $L'=ae^{a\gamma}$ where $\gamma>0$. 
  Taking $n=cT$ with $c<1$, we can  see that $|\greedy(G,T)-nh^{*}(T/n)|\leq o(T)$,

\end{proof}
\subsection{Proof of \Cref{corr:lbstochastic_k=1}}\label{app:proof_llb_stochastic_k=1}
\lbstochastickone*
  The proof is organized as follows: \begin{enumerate}
      \item  Finding the stationary solution of \cref{sysdiffequaZk} for $K=1$.
    \item Proving that the stationary solution is exponentially stable.
    \item Applying an improved version of the Wormald theorem on   $\greedy(G,T)$.
    \item Proving that $\greedy(G,T)$ converges to a function depending on the stationary solution.
  \end{enumerate}

Intuitively $K=1$ means that the maximum budget reached by each node in $U$ is equal to 1. From a technical aspect, supposing $K=1$ reduces 
\cref{sysdiffequaZk_bis} to a system of 2 equations as follows, for $t\in [\frac{1}{n},\frac{T}{n}]$
\begin{align}
\label{reduced_system_k=1}
    \begin{cases}
        \dot{z}_{0}(t)&=-\beta\,z_{0}(t)+\frac{z_{1}(t)}{1-z_{0}(t)}(1-e^{-a(1-z_{0}(t))})\\
    \dot{z}_{1}(t)&=\beta z_{0}(t)-\frac{z_{1}(t)}{1-z_{0}(t)} (1-e^{-a(1-z_{0}(t))})\\
    z_{0}(t)+z_{1}(t)&=1
    \end{cases}
\end{align}
By simplifying \cref{reduced_system_k=1}, we reduce the system to the following equation, 
\begin{equation}
\label{equ:z_0K=1}
    \dot{z}_{0}(t)=-\beta\,z_{0}(t)+1-e^{-a(1-z_{0}(t))}
\end{equation}
The following lemma computes the unique stationary solution of \cref{equ:z_0K=1},
\begin{lemma}
\label{lemma:stationary_sol_k=1}
    The stationary solution of \cref{equ:z_0K=1} is unique and is given for $\beta, a>0$  by, 
    $$z^{*}_0=\frac{1}{\beta}-\frac{1}{a}W\left(\frac{a}{\beta}e^{-a\left(1-\frac{1}{\beta}\right)}\right)$$ 
    where $W$ is the Lambert function. 
\end{lemma}
\begin{proof}
    Let's define $G(z_0)= -\beta\,z_{0}+1-e^{-a(1-z_{0})}$, the stationary solution of \cref{equ:z_0K=1} is the solution of $G(z_0^*)=0$ (the homogeneous equation) with $a>0$ and $\beta>0$, 
\begin{align*}
       G(z_0^*)=0\iff  1-e^{-a(1-z_{0}^*)}=\beta\,z_{0}^* &\iff e^{a\left(\frac{1}{\beta}-z_{0}^*\right)} a\left(\frac{1}{\beta}-z_{0}^*\right)=\frac{a}{\beta}e^{-a\left(1-\frac{1}{\beta}\right)}\\
&\iff a\left(\frac{1}{\beta}-z_{0}^*\right)=W\left(\frac{a}{\beta}e^{-a\left(1-\frac{1}{\beta}\right)}\right) 
\end{align*}
Where $W$ is the Lambert function. So, by rearranging the terms in the last equation, the solution of $G(z_0^*)=0$ is given by, 
    $$z^{*}_0=\frac{1}{\beta}-\frac{1}{a}W\left(\frac{a}{\beta}e^{-a\left(1-\frac{1}{\beta}\right)}\right)$$ 
Let's prove the uniqueness of the stationary solution, $G(0)=1-e^{-a}> 0$ and $G(1)=-\beta< 0$ and we have $\forall~0\leq z_0\leq 1, ~~\frac{\mathrm{d}G(z_0)}{\mathrm{d}z_0}=-(\beta+ae^{-a}e^{-a(1-z_0)})<0$. Thus $G(z_0)=0$ has a unique solution. 

\end{proof}

  The following theorem proves that $z_0^*$ is exponentially stable, meaning that $\forall  t\in [\frac{1}{n},\frac{T}{n}]$, $z_0(t)$ converges to $z_0^*$ with an exponential rate. 
\begin{theorem}
\label{expo_stable_k=1}
  For any $f_0\geq 0$ and $t\in [\frac{1}{n},\frac{T}{n}]$, consider the ordinary differential equation (ODE), $$\begin{cases} \dot{z}_0(t)&=\-\beta z_{0}(t)+1-e^{-a(1-z_{0}(t))}\\
    z_0(1/n)&=f_0\end{cases}$$ 
    
    Thus, it implies that $z_0(t)$ converges to $z^{*}_0$ exponentially.
\end{theorem}
\begin{proof}
    The idea here is to prove that for any perturbation that we add to $z^{*}_0$, this perturbation tends to $0$ when $t$ tends to $+\infty$ exponentially. 

      Let's consider $\epsilon(t): \mathbb{R}\to \mathbb{R}$, a pertubation of the stationary solution $z_0^*$,
    \begin{align*}
        \dot{\epsilon}(t)&= -\beta\,z^{*}_0-\beta\,\epsilon(t)+ 1-e^{-a(1-z^{*}_{0}(t)-\epsilon(t))}\\
        \dot{\epsilon}(t)&=-1+e^{-a(1-z^{*}_{0})}-\beta\,\epsilon(t)+ 1-e^{-a(1-z^{*}_{0}-\epsilon(t))}\\
        \dot{\epsilon}(t)&=e^{-a(1-z^{*}_{0})}\left(1-e^{a\epsilon(t)}\right)-\beta\,\epsilon(t)\\
        \dot{\epsilon}(t)&\leq -a\epsilon(t)\,e^{-a(1-z^{*}_{0})} -\beta\,\epsilon(t) ~~~~~~~~~~~ \text{using that}~~(1-e^{a\epsilon(t)}\leq -a\epsilon(t))\\
        \dot{\epsilon}(t)&\leq\epsilon(t)\underbrace{\left(-a\,e^{-a(1-z^{*}_{0})}-\beta\,\right)}_{\leq 0}
    \end{align*}
    Integrating the last equation , we get, 
    \begin{align}
\ln\left(|\epsilon(t)|\right)- \ln(|\epsilon(0)|)&\leq \left(-a\,e^{-a(1-z^{*}_{0})}-\beta\,\right)\,t\\
        |\epsilon(t)|&\leq |\epsilon(0)|\,\exp\left(-t \left(a\,e^{-a(1-z^{*}_{0})}+\beta\,\right)\right)\\
      |\epsilon(t)|&\leq|f_0-z_0^*|\exp\left(-t\left(ae^{-a\left(1-\frac{1}{\beta}+\frac{1}{a}W\left(\frac{a}{\beta}e^{-a(1-\frac{1}{\beta})}\right)\right)}+\beta\right)\right)\label{eq:152}\\
        |\epsilon(t)|&\leq |f_0-z_0^*|\exp\left(-t \beta\left(1+W(e^{-a(1-\frac{1}{\beta})}\right)\right)\label{eq:153}
    \end{align}

Moving from \cref{eq:152} to \cref{eq:153} is done using $\exp(-W(x))=W(x)/x)$.

Thus $\lim\limits_{t \rightarrow +\infty} \epsilon(t)=0$ exponentially with the following rate $\omega = \beta\left(1+W(e^{-a(1-\frac{1}{\beta})})\right)$. 

\end{proof}
\begin{lemma}
\label{lemma: exponential_stability_full_system_k=1}
 $S_{z_0^*}^1=(z_0^*,z_1^*)=(z_0^*,z_0^*\frac{\beta}{g(z_0^*)})$ is an exponentially stable stationary solution of \cref{reduced_system_k=1}. 
\end{lemma}
\begin{proof}
According to \cref{lemma:stationary_sol_k=1}, $z_0^{*}$ is a stationary solution of \cref{equ:z_0K=1}, this implies that $S_{z_0^*}^1=(z_0^*,z_1^*)=(z_0^*,z_0^*\frac{\beta}{g(z_0^*)})$ is a stationary solution of \cref{reduced_system_k=1}. As previously demonstrated, $\forall t\in [\frac{1}{n},\frac{T}{n}]$, $z_0(t)$ converges to $z_0^*$ exponentially. This implies that \cref{equ:z_0K=1} possesses an exponentially stable stationary solution. Given that \cref{equ:z_0K=1} is a reduced version of \cref{reduced_system_k=1}, we can conclude that $S_{z_0^*}^1$ is an exponentially stable stationary solution for \cref{reduced_system_k=1}.
\end{proof}

With all the essential elements assembled, we are now ready to establish the proof for \cref{corr:lbstochastic_k=1}
\begin{proof}
Let $h^*(T/n)=\int_{1/n}^{T/n} (1-e^{-a(1-z_0^*)})\mathrm{d}\tau=\frac{(T-1)(1-e^{-a(1-z_0^*)})}{n}$, we have,
\begin{align*}
    |\greedy(G,T)-nh^*(T/n)|&=|\greedy(G,T)-nh(T/n)+nh(T/n)-nh^*(T/n)|\\
    &\leq \underbrace{\underset{1\leq t\leq T}{\max}~(|\greedy(G,t)-nh(t/n)|)}_{\leq 3e^{L'T/n}a n^{3/4} ~\text{by \cref{theorem:improved wormald}}}+|nh(T/n)-nh^*(T/n)|
\end{align*}
Let's focus on $A=|nh(T/n)-nh^*(T/n)|$,
\begin{align}
    A&= \left|n\int_{1/n}^{T/n} (1-e^{-a(1-z_0(\tau))}) \mathrm{d}\tau- n\int_{1/n}^{T/n} (1-e^{-a(1-z_0^*)}) \mathrm{d}\tau\right|\\
&= \left|n\int_{1/n}^{T/n} (e^{-a(1-z_0^*)}-e^{-a(1-z_0(\tau))}) \mathrm{d}\tau\right|\\
&= \left|ne^{-a(1-z_0^*)}\int_{1/n}^{T/n} (1-e^{a(z_0(\tau)-z_0^*)}) \mathrm{d}\tau\right| \\
&\leq \left|ne^{-a(1-z_0^*)}\int_{1/n}^{T/n} -a(z_0(\tau)-z_0^*) \mathrm{d}\tau\right| ~~~~~~~(\text{using $1-e^{ax}\leq -ax$})\\
&\leq n a e^{-a(1-z_0^{*})}\int_{1/n}^{T/n} |z_{0}(\tau)-z_0^{*}|\mathrm{d}\tau\label{eq:158}\\
&\leq n a e^{-a(1-z_0^{*})} |f_0-z_0^*|\int_{1/n}^{T/n} \exp\left(-\tau \beta\left(1+W(e^{-a(1-\frac{1}{\beta})}\right)\right)\mathrm{d}\tau \label{eq:159}\\
&\leq \frac{n a e^{-a(1-z_0^{*})} |f_0-z_0^*|}{\beta \left(1+W(e^{-a(1-\frac{1}{\beta})})\right)}\left(e^{-\frac{\beta}{n} \left(1+W(e^{-a(1-\frac{1}{\beta})})\right)}-e^{-\frac{T}{n}\beta \left(1+W(e^{-a(1-\frac{1}{\beta})})\right)}\right)
\end{align}
Moving from \cref{eq:158} to \cref{eq:159} is done using \cref{expo_stable_k=1}. Thus we have, 
\begin{align*}
    |\greedy(G,T)-nh^*(T/n)|\leq 3e^{L'T/n}a n^{3/4} +  \frac{n a e^{-a(1-z_0^{*})} |f_0-z_0^*|}{P}\left(e^{-\frac{1}{n}P}-e^{-\frac{T}{n}P}\right)
\end{align*}

  with $L'=ae^{a\epsilon}$ and $P=\beta \left(1+W(e^{-a(1-\frac{1}{\beta})})\right)$, $\epsilon>0$. 

  Now let's focus on $ A= 3e^{L'T/n}a n^{3/4} +  \frac{n a e^{-a(1-z_0^{*})} |f_0-z_0^*|}{P}\left(e^{-\frac{1}{n}P}-e^{-\frac{T}{n}P}\right) $ and considering that $n=\frac{T}{\alpha \log(T)}$ with $\alpha>0$, we get, 
\begin{align*}
  A&\leq  3e^{L'T/n}a n^{3/4} +  \frac{n a e^{-a(1-z_0^{*})} |f_0-z_0^*|}{P}\left(1-e^{-\frac{T}{n}P}\right) \\
 &= \frac{ae^{-a(1-z_{0}^{*})}|f_0-z_0^*|}{P} (\frac{T}{\alpha\log(T)}-\frac{T^{1-\alpha P}}{\alpha\log(T)})\\
   &+3 a \frac{T^{\frac{3}{4}+\alpha L'}}{(\alpha\log(T))^{\frac{3}{4}}}\\
   &\leq \frac{ae^{-a(1-z_{0}^{*})}|f_0-z_0^*|}{P} \frac{T}{\alpha\log(T)}+3 a \frac{T^{\frac{3}{4}+\alpha L'}}{(\alpha\log(T))^{\frac{3}{4}}}
\end{align*}
Taking $\alpha=\frac{1}{4(P+L')}$, with $z_0^*=\frac{1}{\beta}-\frac{1}{a}W\left(\frac{a}{\beta}e^{-a(1-\frac{1}{\beta})}\right)$ and using the fact that $e^{-W(x)}=\frac{W(x)}{x}$,  we get, 
\begin{align*}
     A&\leq \frac{\beta W\left(e^{-a(1-\frac{1}{\beta})}\right)}{\alpha P}|f_0-z_0^{*}|\frac{T}{\log(T)} +3a(4(L'+P))^{\frac{3}{4}}\frac{T^{\frac{3L'+4P}{4(P+L')}}}{(\log(T))^{\frac{3}{4}}}\\
    &= c_1 \frac{T}{\log(T)}+c_2\frac{T^{\omega'}}{(\log(T))^{\frac{3}{4}}}\\
    &\leq c \frac{T}{(\log(T))^{3/4}}
\end{align*}

where  $c_1 = \frac{4\beta(P+L') W\left(e^{-a(1-\frac{1}{\beta})}\right)}{P}$,  $c_2 =3a(4(L'+P))^{\frac{3}{4}}$, $c=c_1+c_2$,$\omega=\frac{3P+4L'}{4(L'+P)}$ and $\omega'=\frac{3L'+4P}{4(P+L')}$. 

\end{proof}
\subsection{Proof of \cref{prop: lower_bound_CR_general_case}}\label{app:proof_lower_bound_general_case}
\lowerboundCRgeneralcase*
\begin{proof}
According to \cref{budget_dynamic_sto}, we have that for all $u\in U$,
$$b_{u,t}=\min(K,b_{u,t-1}-x_{u,t}+ \eta_t) ~~~~~\text{with}~b_{u,0}=b_0\geq 1$$ 

Which gives, 
\begin{align*}
    \greedy(G,T)=\sum_{u\in U}\sum_{t=1}^{T}x_{u,t}=nb_0+\underbrace{\sum_{u\in U}\sum_{t=1}^T\lE[\eta_t]}_{A_1}-\underbrace{\sum_{u\in U}b_{u,T}}_{A_2}-\underbrace{\sum_{u\in U}\sum_{t=1}^T\lE[\indicator{b_{u,t}=K}\indicator{\eta_t=1}]}_{A_3}
\end{align*}
According to \cref{lemma:Y_approximation}, we have w.h.p $\forall k\geq 0, t\in [T], Y_k(t)=nz_k(t/n)+\cO(n^{3/4})$, let's then compute $A_1$, $A_2$ and $A_3$, 
\begin{align*}
    A_1&=\sum_{u\in U}\sum_{t=1}^T\lE[\eta_t]=\beta\,T\\
    A_2&=\sum_{u\in U}b_{u,T}= n \sum_{k=1}^K kz_k(T/n) + \cO\left(\frac{K(K+1)}{2}n^{3/4}\right)\\
    A_3&= \beta \sum_{t=1}^T z_K(t/n) + \cO(\beta Tn^{-1/4})
\end{align*}
Using the following upper bound on $\opt(G,T)\leq nb_0+\beta T$ and $n=\cO(T)$,  we get that, 
\begin{align}
 \CR^{\rm sto}(\greedy,\cD)\geq \frac{nb_0+\beta T-n\sum_{k=1}^K kz_k(T/n)-\beta\sum_{t=1}^T z_K(t/n)}{nb_0+\beta T}+ \cO(T^{-1/4})
\end{align}
According \cref{theo:stationary_solution_stable_all_k}, $\forall \tau\in [\frac{1}{n}, \frac{T}{n}],(z_0(\tau),\hdots,z_{K}(\tau))$ converges to $\bar{S}_{z_0^*}$ asymptotically, this implies that, 
\begin{align}
\label{lower_bound_CR_stochastique}
    \CR^{\rm sto}(\greedy,\cD)\geq \frac{nb_0 + \beta T - \beta T z_0^*\left(\frac{\beta}{g(z_0^*)}\right)^K- nz_0^*\sum_{k=1}^K k\left(\frac{\beta}{g(z_0^*)}\right)^k}{nb_0 + \beta T} + \cO(T^{-1/4})
\end{align}
From $\sum_{k=0}^{K} z_0^{*} \left(\frac{\beta}{g(z_0^{*})}\right)^{k} =1$, we have that $\left(\frac{\beta}{g(z_0^{*})}\right)^{K}=\frac{g(z_0^{*})}{\beta}-\frac{1}{z_0^{*}}\left(\frac{g(z_0^{*})}{\beta}-1\right)$, this gives, 
\begin{align*}
 \CR^{\rm sto}(\greedy,\cD)&\geq \frac{nb_0 + \beta T - \beta T z_0^*\left(\frac{\beta}{g(z_0^*)}\right)^K- nz_0^*\sum_{k=1}^K k\left(\frac{\beta}{g(z_0^*)}\right)^k}{nb_0 + \beta T} + \cO(T^{-1/4}) \\
 &\geq  \frac{nb_0 + \beta T - \beta T z_0^*\left(\frac{g(z_0^{*})}{\beta}-\frac{1}{z_0^{*}}\left(\frac{g(z_0^{*})}{\beta}-1\right)\right)- nz_0^*\sum_{k=1}^K k\left(\frac{\beta}{g(z_0^*)}\right)^k}{nb_0 + \beta T} \\
 &+ \cO(T^{-1/4})\\
 &\geq \frac{nb_0  +g(z_0^{*})T(1-z_0^{*}) - nz_0^*\sum_{k=1}^K k\left(\frac{\beta}{g(z_0^*)}\right)^k}{nb_0 + \beta T}+ \cO(T^{-1/4}) \\
 &\geq \frac{nb_0  +g(z_0^{*})T(1-z_0^{*}) - nz_0^*\sum_{k=1}^K k\left(\frac{\beta}{g(z_0^*)}\right)^k}{nb_0 + \beta T}+ \cO(T^{-1/4})\\
\end{align*}
Using $1-\left(\frac{\beta}{g(z_0^{*})}\right)^{K+1}=\frac{1}{z_0^{*}}\left(1-\frac{\beta}{g(z_0^{*})}\right)$ and $\sum_{k=1}^{K}k x^{k}= x \frac{\mathrm{d}}{\mathrm{d} x}\left(\frac{1-x^{K+1}}{1-x}\right)$ with $x=\frac{\beta}{g(z_0^{*})}$, 
\begin{align*}
     \CR^{\rm sto}(\greedy,\cD)&\geq\frac{nb_0  +g(z_0^{*})T(1-z_0^{*}) - nz_0^*\sum_{k=1}^K k\left(\frac{\beta}{g(z_0^*)}\right)^k}{nb_0 + \beta T}+ \cO(T^{-1/4})\\
     &\geq\frac{nb_0  +g(z_0^{*})T(1-z_0^{*}) - n\left(\frac{\beta}{g(z_0^{*})-\beta }-\frac{(K+1)\beta^{K+1}}{g(z_0^*)^{K+1}-\beta^{K+1}}\right)}{nb_0 + \beta T}+ \cO(T^{-1/4})
\end{align*}
\end{proof}
\subsection{Proof of \cref{theo:lb_stochastic_k=1_CR}}\label{app:proof_lbstochastickoneCR}
\lbstochastickoneCR*
\begin{proof} 
 $$\underset{T \to +\infty}{\lim} \CR^{\rm sto}(\greedy,\cD)= \frac{g(z_0^{*})(1-z_0^{*})}{\beta}$$
When $K\to \infty$, $z_0^{*}$ satisfies $\sum_{k=0}^{+\infty} z_0^{*} \left(\frac{\beta}{g(z_0^{*})}\right)^{k}=1$, this gives, 
$$\sum_{k=0}^{+\infty} z_0^{*} \left(\frac{\beta}{g(z_0^{*})}\right)^{k}= z_0^{*}\frac{1}{1-\frac{\beta}{g(z_0^{*})}}=1  ~~\Longrightarrow~~ 1-e^{-a(1-z_0^{*})}=\beta$$
 which leads to $z_0^{*}= 1+\frac{ln(1-\beta)}{a}$. 

Thus,
 $$\underset{K,n \to +\infty}{\lim} ~\underset{T \to +\infty}{\lim} \CR^{\rm sto}(\greedy,\cD)= \underset{K,n \to +\infty}{\lim}  \frac{g(z_0^{*})(1-z_0^{*})}{\beta} =1$$

\end{proof}

\end{document}